\documentclass[11pt]{article}
\usepackage{natbib}
\usepackage{amsmath}
\usepackage{amsfonts,amssymb,amsthm}
\usepackage{setspace}
\usepackage{upgreek}
\usepackage{float}
\usepackage{graphicx}
\DeclareMathOperator*{\argmin}{\arg\min}
\theoremstyle{definition}
\newtheorem{defn}{Definition}[section]

\theoremstyle{plain}

\newtheorem{thm}[defn]{Theorem}
\newtheorem{cor}[defn]{Corollary}

\newtheorem{lem}[defn]{Lemma}
\theoremstyle{remark}

\newtheorem{rem}[defn]{Remark}

\newtheorem{exs}[defn]{Examples}

\numberwithin{equation}{section}

\usepackage[colorinlistoftodos]{todonotes}
\usepackage[colorlinks=true, allcolors=blue]{hyperref}

\newcommand\me{\mathrm{e}}
\newcommand\mb{\mathbb}
\newcommand\mE{\mathbb{E}}
\newcommand\mP{\mathbb{P}}
\newcommand\mF{\mathcal{F}}

\usepackage[margin=1in]{geometry}
\begin{document}



\title{\textbf{High-dimensional Simultaneous Inference on Non-Gaussian VAR Model via De-biased Estimator}}
\author{Linbo Liu and Danna Zhang\\Department of Mathematics, UC San Diego}
\date{October 26, 2021}
\maketitle
\begin{abstract}
	Simultaneous inference for high-dimensional non-Gaussian time series is always considered to be a challenging problem. Such tasks require not only robust estimation of the coefficients in the random process, but also deriving limiting distribution for a sum of dependent variables. In this paper, we propose a multiplier bootstrap procedure to conduct simultaneous inference for the transition coefficients in high-dimensional non-Gaussian vector autoregressive (VAR) models. This bootstrap-assisted procedure allows the dimension of the time series to grow exponentially fast in the number of observations.  As a test statistic, a de-biased estimator is constructed for simultaneous inference. Unlike the traditional de-biased/de-sparsifying Lasso estimator, robust convex loss function and normalizing weight function are exploited to avoid any unfavorable behavior at the tail of the distribution. 
	We develop Gaussian approximation theory for VAR model to derive the asymptotic distribution of the de-biased estimator and propose a multiplier bootstrap-assisted procedure to obtain critical values under very mild moment conditions on the innovations. 
	As an important tool in the convergence analysis of various estimators,  we establish a Bernstein-type probabilistic concentration inequality for bounded VAR models. Numerical experiments verify the validity and efficiency of the proposed method.
\end{abstract}

\section{Introduction}
High-dimensional statistics become increasingly important due to the rapid development of information technology in the past decade. In this paper, we are primarily interested in conducting simultaneous inference via de-biased $M$-estimator on the transition matrices in a high-dimensional vector autoregressive model with non-Gaussian innovations. An extensive body of work has been proposed on estimation and inference on the coefficient vector in linear regression setting and we refer readers to \cite{buhlmann2011statistics} for an overview of recent development in high-dimensional statistical techniques. $M$-estimator is one of the most popular tools among them, which has been proved a success in signal estimation (\cite{negahban2012unified}), support recovery (\cite{loh2017support}), variable selection (\cite{zou2006adaptive}) and robust estimation with heavy-tailed noises using nonconvex loss functions (\cite{loh2017statistical}). As a penalized $M$-estimator, Lasso (\cite{tibshirani1996regression}) also plays an important role in estimating transition coefficients in high-dimensional 
VAR models beyond linear regression; see for example \cite{hsu2008subset}, \cite{nardi2011autoregressive}, \cite{basu2015regularized} among others. Another line of work is to achieve such estimation tasks by Dantzig selector (\cite{candes2007dantzig}). \cite{han2015direct} proposed a new approach to estimating the transition matrix via Dantzig-type estimator and solved a linear programming problem. They remarked that this estimation procedure enjoys many advantages including computational efficiency and weaker assumptions on the transition matrix. However, the aforementioned literature mainly discussed the scenario where Gaussian or sub-Gaussian noises are in presence. 

To deal with the heavy-tailed errors, regularized robust methods have been widely studied. For instance, \cite{li20081} proposed an $\ell_1$-regularized quantile regression method in low dimensional setting and devised an algorithm to efficiently solve the proposed optimization problem. \cite{wu2009variable} studied penalized quantile regression from the perspective of variable selection. However, quantile regression and least absolute deviation regression can be significantly different from the mean function, especially when the distribution of noise is asymmetric. To overcome this issue, \cite{fan2017estimation} developed robust approximation Lasso (RA-Lasso) estimator based on penalized Huber loss and proved the feasibility of RA-Lasso in estimation of high-dimensional mean regression. Apart from linear regression setting, \cite{zhang2019robust} also used Huber loss to obtain a consistent estimate of mean vector and covariance matrix for high-dimensional time series. Also, robust estimation of the transition coefficients was studied in \cite{liu2021robust} via two types of approaches: Lasso-based and Dantzig-based estimator.

Besides estimation, recent research effort also turned to high-dimensional statistical inference, such as performing multiple hypothesis testing and constructing simultaneous confidence intervals, both for regression coefficients and mean vectors of random processes. To tackle the high dimensionality, the idea of low dimensional projection was exploited by numerous popular literature. For instance, \cite{javanmard2014confidence}, \cite{van2014asymptotically}, \cite{zhang2014confidence} constructed de-sparsifying Lasso by inverting the Karush-Kuhn-Tucker (KKT) condition and derived asymptotic distribution for the projection of high-dimensional parameters onto fixed-dimensional space. As an extension of the previous techniques, \cite{loh2018scale} proposed the asymptotic theory of one-step estimator, allowing the presence of non-Gaussian noises. Employing Gaussian approximation theory (\cite{chernozhukov2013gaussian}), \cite{zhang2017simultaneous} proposed a bootstrap-assisted procedure to conduct simultaneous statistical inference, which allowed the number of testing to greatly surpass the number of observations as a significant improvement. Although a huge body of work has been completed for the inference of regression coefficients, there have been limited research on the generalization of these theoretical properties to time series, perhaps due to the technical difficulty when generalizing Gaussian approximation results to dependent random variables. \cite{zhang2017gaussian} adopted the framework of functional dependence measures (\cite{wu2005nonlinear}) to account for temporal dependency and provided Gaussian approximation results for general time series. They also showed, as an application, how to construct simultaneous confidence intervals for mean vectors of high-dimensional random processes with asymptotically correct coverage probabilities. 

In this paper, we consider simultaneous inference of transition coefficients in possibly non-Gaussian vector autoregressive (VAR) models with lag $d$:
\begin{equation*}
	X_i = A_1X_{i-1}+A_2X_{i-2}+\dots+A_dX_{i-d}+\varepsilon_i,\quad i=1,\dots,n,
\end{equation*}
where $X_i\in\mb{R}^p$ is the time series, $A_i\in\mb{R}^{p\times p},\,i=1,\dots,d$ are the transition matrices, and $\varepsilon_i\in\mb{R}^p$ are the innovation vectors. We allow the dimension $p$ to exceed the number of observations $n$, or even $\log p=o(n^b)$ for some $b>0$, as is commonly assumed in high-dimensional regime. Different from many other work, we do not impose Gaussianity or sub-Gaussianity assumptions on the noise terms $\varepsilon_i$. 

We are particularly interested in the following simultaneous testing problem:
$$H_0:A_i=A_i^0,\quad \text{for all }i=1,\dots, d$$
versus the alternative hypothesis 
$$H_1:A_i\neq A_i^0,\quad\text{for some }i=1,\dots,d.$$
It's worth mentioning that the above problems still have $p^2$ null hypotheses to verify even if the lag $d=1$. We propose to build a de-biased estimator $\check\beta$ from some consistent pilot estimator $\widehat\beta$ (for example, the one provided in \cite{liu2021robust}). There are a few challenges when we prove the feasibility of de-biased estimator as well as its theoretical guarantees: (i) VAR models display temporal dependency across observations, which makes the majority of probabilistic tools such as classic Bernstein inequality and Gaussian approximation inapplicable. (ii) Fat-tailed innovations $\varepsilon_i$ imply fat-tailed $x_i$ in VAR model, while robust methods regarding linear regression can assume $\varepsilon_i$ to have heavy-tail but $x_i$ remains sub-Gaussian (\cite{fan2017estimation} and \cite{zhang2017simultaneous}). (iii) We hope our simultaneous inference procedure to work in ultra-high dimensional regime, where $p$ can grow exponentially fast in $n$. As a result, these challenges inspire us to establish a new Bernstein-type inequality (section \ref{estimation}) and Gaussian approximation  results (section \ref{GA}) under the framework of VAR model. Also, we will adopt the definition of spectral decay index to capture the dependency among time series data, as in \cite{liu2021robust}.

The paper is organized as follows. In section \ref{main}, we first present more details and some preparatory definitions of VAR models and propose the test statistics for simultaneous inference via de-biased estimator, which is constructed through a robust loss function and a weight function on $x_i$. The main result delivering critical values for such test statistics by multiplier bootstrap is given in section \ref{thm}. In section \ref{estimation}, we complete the estimation of multiple statistics by establishing a Bernstein inequality. A thorough discussion of Gaussian approximation and its derivation under VAR model are presented in section \ref{GA}. Some numerical experiments are conducted in section \ref{num} to assess the empirical performance of the multiplier bootstrap procedure.

Finally, we introduce some notation. For a vector $\beta=(\beta_1,\dots,\beta_p)^\top$, let $|\beta|_1=\sum_{i}|\beta_i|$, $|\beta|_2=(\,{\sum_{i}\beta_i^2}\,)^{1/2}$ and  $|\beta|_\infty=\max_{i}|\beta_i|$ be its $\ell_1,\ell_2,\ell_\infty$ norm respectively. For a matrix $A=(a_{ij})_{1\leq i,j\leq p}$
, let $\lambda_i,\,i=1,\dots,p$, be its eigenvalues and $\lambda_{\max}(A),\,\lambda_{\min}(A)$ be its maximum and minimum eigenvalues respectively. Also let $\rho(A)=\max_{i}|\lambda_i|$ be the spectral radius. Denote $\|A\|_1 = \max_j \sum_i|a_{ij}|$, $\|A\|_\infty = \max_i\sum_{j} |a_{ij}|$, and spectral norm $\|A\| = \|A\|_2 = \sup_{|x|_2 \neq 0}|Ax|_2/|x|_2$. Moreover, let $\|A\|_{\max}=\max_{i,j}|a_{ij}|$ be the entry-wise maximum norm.  For a random variable $X$ and $q>0$, define $\|X\|_q=(\mE[X^q])^{1/q}$. For two real numbers $x,y$, set $x\lor y=\max(x,y)$. For two sequences of positive numbers $\{a_n\}$ and $\{b_n\}$, we write $a_n\lesssim b_n$ if there exists some constant $C>0$, such that $a_n/b_n\leq C$ as $n\to\infty$, and also write $a_n\asymp b_n$ if $a_n\lesssim b_n$ and $b_n\lesssim a_n$. We use $c_0,c_1,\dots$ and $C_0,C_1,\dots$ to denote some universal positive constants whose values may vary in different context. Throughout the paper, we consider the high-dimensional regime, allowing the dimension $p$ to grow with the sample size $n$, that is, we assume $p=p_n\to\infty$ as $n\to\infty$.


\section{Main Results}\label{main}
\subsection{Vector autoregressive model}
Consider a VAR(d) model:
\begin{equation}
	X_i = A_1X_{i-1}+A_2X_{i-2}+\dots+A_dX_{i-d}+\varepsilon_i,\quad i=1,\dots,n,
\end{equation}
where $X_i=(X_{i1}, X_{i2}, \dots, X_{ip})\in\mb{R}^p$ is the random process of interests, $A_i\in\mb{R}^{p\times p}$, $i=1,\dots,d$, are the transition matrices and $\varepsilon_i$, $i\in\mb{Z}$, are i.i.d. innovation vectors with zero mean and symmetric distribution, i.e. $\varepsilon_i=-\varepsilon_i$ in distribution, for all $i\in\mb{Z}$. By a rearrangement of variables, VAR(d) models can be formulated as VAR(1) models (see \cite{liu2021robust}). Therefore, without loss of generality, we shall work with VAR(1) models:
\begin{equation}\label{var1}
	X_i = AX_{i-1}+\varepsilon_i,\quad i=1,\dots,n.
\end{equation}
This type of random process has a wide range of application, such as finance development (\cite{shan2005does}), economy (\cite{juselius2006cointegrated}) and exchange rate dynamics (\cite{wu2010var}).

To ensure model stationarity, we assume that the spectral radius $\rho(A)<1$ throughout the paper, which is also the sufficient and necessary condition for a VAR(1) model to be stationary. However, a more restrictive condition that $\|A\|<1$ is always assumed in most of the earlier work. See for example, \cite{han2015direct}, \cite{loh2012high} and \cite{negahban2011estimation}. For a non-symmetric matrix $A$, it could happen that $\|A\|\geq1$ while $\rho(A)<1$. To fill the gap between $\rho(A)$ and $\|A\|$, \cite{basu2015regularized} proposed stability measures for high-dimensional time series to capture temporal and cross-section dependence via the spectral density function 
$$f_X(\theta)=\frac1{2\pi}\sum_{\ell=-\infty}^\infty\Gamma_X(\ell)\me^{-i\ell\theta},\quad\theta\in[-\pi,\pi],$$
where $\Gamma_X(\ell)=\text{Cov}(X_i, X_{i+\ell}),~i,\ell\in\mb{Z}$ is the autocovariance function of the process $\{X_i\}$. In a more recent work, \cite{liu2021robust} defined spectral decay index to connect $\rho(A)$ with $\|A\|$ from a different point of view. In this paper, we will adopt the framework of spectral decay index in \cite{liu2021robust}.
\begin{defn}\label{defn2.1}
	For any matrix $A\in\mb{R}^{p\times p}$ such that $\rho(A)<1$, define the \emph{spectral decay index} as 
	\begin{equation}\label{tau}
	\tau=\min\{t\in\mb{Z}^+: \|A^t\|_{\infty}<\rho\}
	\end{equation}
	 for some constant $0<\rho<1$.
\end{defn}
\begin{rem}
	Note that in \eqref{tau}, we use $L_\infty$ norm, while spectral norm is considered in \cite{liu2021robust}. However, the spectral decay index shares many properties even if defined in different matrix norms. Some of them are summarized as follows. For any matrix $A$ with $\rho(A)<1$, finite spectral decay index $\tau$ exists. 
	In general, $\tau$ may not be of constant order when the dimension $p$ increases. Technically speaking, we need to explicitly write $\tau=\tau_p$ to capture the dependence on $p$. However, in the rest of the paper, we simply write $\tau$ for ease of notation. For more analysis of spectral decay index, see section 2 of \cite{liu2021robust}.
\end{rem}

Next, we are interested in building some estimators of $A$ for which we could establish asymptotic distribution theory. This allows one to conduct statistical inference, such as finding simultaneous confidence interval. There have been some work on the robust estimation only. \cite{liu2021robust} provides both a Lasso-type estimator and a Dantzig-type estimator to consistently estimate the transition coefficient $A$ given $\{X_i\}$, under very mild moment condition on $X_i$ and $\epsilon_i$. It turns out that both Lasso-type and Dantzig-type estimators are not unbiased for estimating the transition matrix, thus insufficient for tasks like statistical inference. Therefore, one needs to develop more refined method to establish results in terms of asymptotic distributional theory. In the following sections, we will construct a de-biased estimator based on the existing one and derive the limiting distribution for the de-biased estimator.

Note that unlike many other existing work (\cite{han2015direct}, \cite{basu2015regularized}, etc.), we do not assume $\varepsilon_i$ to be Gaussian or sub-Gaussian. Instead, it could happen that the innovations $\varepsilon_i$ only have some finite moments, which makes the standard techniques for estimation and inference invalid.

\subsection{De-biased estimator}
In this section, we construct a de-biased estimator using the techniques introduced in \cite{bickel1975one}. To fix the idea, let $a_j^\top$ be the $j$-th row of $A$ and $\beta^*=\text{Vec}(A)=(a_1^\top,a_2^\top,\dots,a_p^\top)^\top\in\mb{R}^{p^2}$. Suppose we are given a consistent, possibly biased, estimator $\widehat\beta$ of $\beta^*$, i.e. $|\widehat\beta-\beta^*|=o(1)$ (for example, Lasso-type or Dantzig-type estimators in \cite{liu2021robust}).
 Define a loss function $L:\mb{R}^{p^2}\to\mb{R}$ as
\begin{equation}\label{eq2.3}
	L_n(\beta)=\frac1{n}\sum_{i=1}^{n}\sum_{k=1}^p\ell(X_{ik}-X_{i-1}^\top\beta_k)w(X_{i-1}),
\end{equation}
where $\beta=(\beta_1^\top,\dots,\beta_p^\top)^\top$ with $\beta_k\in\mb{R}^p$ for $1\leq k\leq p$, the weight function
$$w(x)=\min\bigg\{1,\frac{T^3}{|x|_\infty^3}\bigg\}$$ 
for some threshold $T>0$ to be determined later, and the robust loss function $\ell(x)$ satisfies:
	\begin{itemize}
		\item[(i)]  $\ell(x)$ is a thrice differentiable convex and even function.
		\item[(ii)] For some constant $C>0$, $|\ell'|,|\ell''|,|\ell^{(3)}|\leq C$.
	\end{itemize}
We give two examples of such loss functions from \cite{pan2021iteratively} that satisfy the above conditions.
\begin{exs}[Smoothed huber loss I]
\begin{equation*}
	\ell(x)=\begin{cases}
		x^2/2-|x|^3/6\quad &\text{if }|x|\leq1,\\
		|x|/2 - 1/6\quad &\text{if }|x|>1.
	\end{cases}
\end{equation*}
\end{exs}

\begin{exs}[Smoothed huber loss II]
	\begin{equation*}
		\ell(x)=\begin{cases}
			x^2/2-x^4/24,\quad &\text{if }|x|\leq\sqrt2,\\
			\big(2\sqrt2/3\big)|x|-1/2,\quad&\text{if }|x|>\sqrt2.
		\end{cases}
	\end{equation*}
\end{exs}
Direct calculation shows that everywhere twice differentiable and almost everywhere thrice differentiable. Also, the derivative of first three orders are bounded in magnitude. We mention that generalization to other loss functions that does not satisfy the differentiability conditions (for example, huber loss) may be derived under more refined arguments, but will be omitted in this paper. 

Denote by $\psi(x)=\ell'(x)$ the derivative of $\ell(x)$, then $\psi(x)$ is twice differentiable by condition (i) and $|\psi(x)|\leq C$ for all $x\in\mb{R}$ by condition (ii). Let $\mu=(\mu_1,\dots,\mu_p)^\top\in\mb{R}^p$ with $\mu_k=\mE[\psi'(\varepsilon_{ik})]$ and $\mu^{-1}=(\mu_1^{-1},\dots,\mu_p^{-1})^\top$. 
Let $\widehat\mu=(\widehat\mu_1,\dots,\widehat\mu_p)$ be the estimate of $\mu$ with $\widehat\mu_k=\frac1n\sum_{i=1}^n\psi'(\widehat\varepsilon_{ik})$, where $\widehat\varepsilon_{ik}=X_{ik}-X_{i-1}^\top\widehat\beta_k$. 
Let $\Sigma_x=\mE[X_{i}X_{i}^\top w(X_{i})]\in\mb{R}^{p\times p}$ be the weighted covariance matrix and $\Omega_x=\Sigma_x^{-1}\in\mb{R}^{p\times p}$ be the weighted precision matrix. Denote by $\widehat\Sigma_x=n^{-1}\sum_{i=1}^nX_{i-1}X_{i-1}^\top w(X_{i-1})$ the weighted sample covariance. Furthermore, suppose that $\widehat\Omega_x$ is a suitable approximation of the weighted precision matrix $\Omega_x$ (e.g., CLIME estimator introduced by \cite{cai2011constrained}), as will be discussed in section 2.3. To ensure the validity of such estimator, the sparsity of each row of $\Omega_x$ is always assumed due to high dimensionality.
Now we introduce a few more notations:
	\begin{align}\label{eq2.5}
		\Sigma&=\text{diag}(\mu)\otimes\Sigma_x=\begin{bmatrix}
			\mu_1\Sigma_x & 0 & 0 & \dots & 0\\
			0 & \mu_2\Sigma_x & 0 & \dots & 0\\
			\vdots & \vdots & \ddots & \dots & 0\\
			0 & 0& 0 & 0 & \mu_p\Sigma_x
		\end{bmatrix}\in\mb{R}^{p^2\times p^2},
	\end{align}
and analogously, 
\begin{align}
	\Omega&=\Sigma^{-1}=\text{diag}(\mu^{-1})\otimes\Omega_x;\notag\\
	\widehat\Sigma&=\text{diag}(\widehat\mu)\otimes\widehat\Sigma_x,\quad \widehat\Omega=\text{diag}(\widehat\mu^{-1})\otimes\widehat\Omega_x.
\end{align}
 Following the one-step estimator in \cite{bickel1975one}, we de-bias $\widehat\beta$ by adding an additional term involving the gradient of the loss function $L$:
\begin{equation}\label{eq2.6}
	\check\beta=\widehat\beta+\widehat\Omega\,\nabla L_n(\widehat\beta).
\end{equation}
To briefly explain the presence of $\widehat\Omega$, consider Taylor expansion of $\nabla L_n(\widehat\beta)$ around $\nabla L_n(\beta^*)$. Write
\begin{align}\label{eq2.7}
	\sqrt n(\check\beta-\beta^*)&=\sqrt n(\widehat\beta-\beta^*)+\sqrt n\,\widehat\Omega\,\nabla L_n(\beta^*)-\sqrt n\,\widehat\Omega(\nabla L_n(\widehat\beta)-\nabla L_n(\beta^*))\notag\\
	&=\sqrt n\,\widehat\Omega\,\nabla L_n(\beta^*)+\sqrt n\left[(\widehat\beta-\beta^*)-\widehat\Omega\,\nabla^2 L_n(\beta^*)(\widehat\beta-\beta^*)+ R\right]\notag\\
	&=\underbrace{\sqrt n \,\widehat\Omega\,\nabla L_n(\beta^*)}_A+\underbrace{\sqrt n\left[\Big(I_{p^2}-\widehat\Omega\,\nabla^2L_n(\beta^*)(\widehat\beta-\beta^*)\right]}_\Delta+\sqrt nR,
\end{align}
where the remainder term $\sqrt nR=o(1)$ under certain conditions. Moreover, we also hope $\Delta$ to be negligible. As will be shown in the following sections,
\begin{equation}\label{eq2.8}
	\Delta\leq\sqrt n\Big(\|\Omega-\widehat\Omega\|_{1}\|\Sigma\|_{\max}+\|\nabla^2L_n(\beta^*)-\Sigma\|_{\max}\|\widehat\Omega\|_1\Big)|\widehat\beta-\beta^*|_1,
\end{equation}

To this end, $\widehat\Omega$ needs to be a good approximation of the precision matrix $\Omega$, which inspires the construction of such $\widehat\Omega$. More rigorous arguments will be presented in the subsequent sections.

Note that the estimator $\check\beta$ is closely related to the de-sparsifying Lasso estimator (\cite{van2014asymptotically} and \cite{zhang2014confidence}), which is employed to conduct simultaneous inference for linear regression models in \cite{zhang2017simultaneous}. $\check\beta$ will reduce to de-sparsifying Lasso estimator if the loss $\ell(x)$ in \eqref{eq2.3} is squared error loss and the weight $w(x)\equiv1$. Moreover, \cite{loh2018scale} uses this one-step estimator to build the limiting distribution of high-dimensional vector restricted to a fixed number of coordinates, and delivers a result that agrees with \cite{bickel1975one} for low-dimensional robust M-estimators. Different from that, we will derive such conclusions simultaneously for all $p^2$ coordinates of $\beta^*$.

In the subsequent sections, we aim at obtaining a limiting distribution for $\check\beta$.

\subsection{Estimation of the precision matrix}
In this section, we mainly discuss the validity of having $\widehat\Omega$ as an approximation of $\Omega$. By the structure of $\Omega$, we need to first find a suitable estimator of the weighted precision $\Omega_x$. 

The estimation of the sparse inverse covariance matrix based on a collection of observations $\{X_i\}$ plays a crucial role in establishing the asymptotic distribution. In high-dimensional regime, one cannot obtain a suitable estimator for the precision matrix by simply inverting the sample covariance, as the sample covariance is not invertible when the number of features exceeds the number of observations. Depending on the purposes, various methodology have been proposed to solve problem of estimating the precision. See for example, graphical Lasso (\cite{yuan2007model} and \cite{friedman2008sparse}) and nodewise regression (\cite{meinshausen2006high}). From a different perspective, \cite{cai2011constrained} proposed a CLIME approach to sparse precision estimation, which shall be applied in this paper. For completeness, we reproduce the CLIME estimator in the following.

Suppose that the sparsity of each row of $\Omega_x$ is at most $s$, i.e., $s=\max_{1\leq i\leq p}|\{j:\Omega_{x,ij}\neq0\}|.$
We first obtain $\widehat\Theta$ by solving
	$$\widehat\Theta=\argmin_{\Theta}\sum_{i,j}\big|\Theta_{ij}\big|\quad\text{subject to: }\quad \|\widehat\Sigma_x\Theta-I_p\|_{\max}\leq \lambda_n,$$
for some regularization parameter $\lambda_n>0$. Note that the solution $\Theta$ may not symmetric. To account for symmetry, the CLIME estimator $\widehat\Omega_x$ is defined as
\begin{equation}\label{clime}\widehat\Omega_x=(\widehat\omega_{ij}),\quad\text{where }\widehat\omega_{ij}=\widehat\omega_{ji}=\widehat\Theta_{ij}\mathbb{I}\{|\widehat\Theta_{ij}|\leq|\widehat\Theta_{ji}|\}+\widehat\Theta_{ji}\mathbb{I}\{|\widehat\Theta_{ij}|>|\widehat\Theta_{ji}|\}.\end{equation}
For more analysis of CLIME estimator, see \cite{cai2011constrained}. Next, we present the convergence theorem for CLIME estimator.

\begin{thm}\label{thm4.1}
	Let $\tau$ be defined in definition \ref{defn2.1} and $\gamma=\max_{t=0,1\dots,\tau-1}\|A^t\|$.
	Choose $\lambda_n\asymp\|\Omega_x\|_1\gamma \tau^2 T^2(\log p)^{3/2}n^{-1/2}$, then with probability at least $1-4p^{-c_0}$ for some constant $c_0>0$,
	$$\|\widehat\Omega_x-\Omega_x\|_{\max}\lesssim\|\Omega_x\|_1\lambda_n\quad\text{and}\quad \|\widehat\Omega_x-\Omega_x\|_{1}\lesssim \|\Omega_x\|_1s\lambda_n.$$
\end{thm}
\begin{rem}
Theorem \ref{thm4.1} is a direct application of Theorem 6 of \cite{cai2011constrained}. Note that if we assume the eigenvalue condition on $\Sigma_x$ that $0\leq c\leq\lambda_{\min}(\Sigma_x)\leq\lambda_{\max}(\Sigma_x)\leq C$, then $\|\Omega_x\|_2\leq1/\lambda_{\min}(\Sigma_x)=O(1).$ Therefore, by the sparsity condition on $\Omega_x$, we immediately have that $\|\Omega_x\|_1=O(\sqrt s).$ Suppose the scaling condition holds that $s\gamma\tau^2 T^2(\log p)^{3/2}n^{-1/2}=o(1)$, then the CLIME estimator $\widehat\Omega_x$ defined in \eqref{clime} is consistent in estimating the weighted precision matrix of the VAR(1) model \eqref{var1}. 
\end{rem}
The following theorem shows that $\|\Omega-\widehat\Omega\|$ enjoys the same convergence rate as in the previous theorem.
\begin{thm}\label{thm4.2}
	Let $\widehat\Omega_x$ be the CLIME estimator defined above. Assume that $\mu_k>c_1>0$ for all $1\leq k\leq p$, then with probability at least $1-6p^{-c}$,
	$$\|\Omega-\widehat\Omega\|_{\max}\lesssim \|\Omega_x\|_{1}\lambda_n\quad\text{and}\quad\|\Omega-\widehat\Omega\|_{1}\lesssim \|\Omega_x\|_{1}s\lambda_n.$$
\end{thm}
The above theorem is built upon two facts: $\widehat\Omega_x$ approximates $\Omega_x$ and $\widehat\mu$ approximates $\mu$. The result regarding the latter approximation will be given in Lemma \ref{lem1.1}.

\subsection{Simultaneous inference}\label{thm}
In this section, consider the following hypothesis testing problem:
$$H_0:A_{ij}=A^0_{ij},\quad\text{for all }i, j=1,\dots, p\quad(\text{or equivalently, }\beta^*_j=\beta^0_j,\quad\text{for all }j=1,\dots, p)$$
versus the alternative hypothesis $H_1:A_{ij}\neq A^0_{ij}$ for some $i, j=1,\dots,p$. Instead of projecting the explanatory variables onto a subspace of fixed dimension (\cite{javanmard2014confidence}, \cite{zhang2014confidence}, \cite{van2014asymptotically} and \cite{loh2018scale}), we allow the number of testings to grow as fast as an exponential order of the sample size $n$. \cite{zhang2017simultaneous} presented a more related work, where it's also allowed that the testing size to grow as a function of $p$. However, they conducted such simultaneous inference procedure under linear regression setting with independent random variables.

Employing the de-biased estimator $\check\beta$ defined in \eqref{eq2.6}, we propose to use the test statistics 
\begin{equation}\label{eq2.9}
	\sqrt n|\check\beta-\beta^0|_{\infty},
\end{equation}
where $\check\beta$ is defined in \eqref{eq2.6}. In the next several theorems, we elaborate a multiplier bootstrap method to obtain the critical value of the test statistics, which requires a few scaling and moment assumptions. Recall definition \ref{defn2.1} for $\tau$ and theorem \ref{thm4.1} for the definition of $\gamma$. Also recall that $s=\max_{1\leq i\leq p}|\{j:\Omega_{x,ij}\neq0\}|.$ 

\textbf{Assumptions}
\begin{itemize}
	\item[(A1)] $\sqrt nT^3|\widehat\beta-\beta^*|_1^2=o(1)$.
	\item[(A2)] $\|\Omega_x\|_1^2s\gamma^2\tau^4T^4(\log p)^3/\sqrt n=o(1)$.
	\item[(A3)] $s\gamma\tau^2T^2(\log p)^{3/2}|\widehat\beta-\beta^*|_1=o(1)$.
	\item[(A4)] $sT^2(\log(pn))^7/n\lesssim n^{-c}$.
	\item[(A5)] $(\log p)^{3/2}(\log n)^{1/2}T\sqrt{s\tau}\gamma/n^{1/4}=o(1).$
\end{itemize}
Additionally, throughout the paper we assume that for some constant $C>0$, $\mE[X_{ik}^2]\leq C$ and $\mE[\varepsilon_{ik}^2]\leq C$ for all $1\leq k\leq p$. We also suppose that $\|\Sigma_x\|_{\max}=O(1)$ and $0<c\leq\lambda_{\min}(\Sigma_x)\leq\lambda_{\max}(\Sigma_x)\leq C$. Thus, $\|\Omega_x\|_2\leq1/\lambda_{\min}(\Sigma_x)=O(1)$ and $\|\Omega_x\|_1=O(\sqrt s)$, where the row sparsity $s=\max_{1\leq i\leq p}|\{j:\Omega_{x,ij}\neq0\}|$.
\begin{thm}\label{thm2.1}
	Let $\zeta_1=\gamma\tau^2T^2(\log p)^{3/2}|\widehat\beta-\beta^*|_1+\sqrt nT^3|\widehat\beta-\beta^*|_1^2+s\gamma^2\tau^4T^4(\log p)^3/\sqrt n.$
	Suppose assumptions (A1) --- (A3) hold. Additionally assume that $\zeta_1\sqrt{1\lor\log(p/\zeta_1)}=o(1)$, then  we have that 
	$$\mP\bigg(|\sqrt n(\check\beta-\beta^*)-\sqrt n\Omega\nabla L_n(\beta^*)|_\infty>\zeta_1\bigg)<\zeta_2,$$
	where $\zeta_1\sqrt{1\lor\log(p/\zeta_1)}=o(1)$ and $\zeta_2=o(1)$.
\end{thm}
Theorem \ref{thm2.1} rigorously verifies that $\sqrt nR=o(1)$ and $\Delta=o(1)$ in  \eqref{eq2.7} by the proposed construction of $\widehat\Omega$ and suggests us to perform further analysis on $\sqrt n\Omega\nabla L_n(\beta^*)$. To derive the limiting distribution, we shall use Gaussian approximation technique, since the classic central limit theorem fails in high-dimensional setting. 

Gaussian approximation was initially invented for high-dimensional independent random variables in \cite{chernozhukov2013gaussian} and further generalized to high-dimensional time series in 
	\cite{zhang2017gaussian}. \cite{zhang2017simultaneous} and \cite{loh2018scale} applied the GA technique in \cite{chernozhukov2013gaussian} to the derivation of asymptotic distribution in linear regression setting. However, data generated from VAR model suffers temporal dependence, which makes the aforementioned techniques unavailable. Although \cite{zhang2017gaussian} established such GA results for general time series using dependence adjusted norm, direct application of their theorems does not yield desirable conclusion in ultra-high dimensional setting. This leads us to derive a new GA theorem with better convergence rate, which is achievable thanks to the structure of VAR model.

The next theorem establishes a Gaussian approximation(GA) result for the term $\sqrt n\Omega\nabla L_n(\beta^*)$. For a more detailed description of Gaussian approximation procedure, see section XXX.
\begin{thm}\label{thm2.2}
	Denote $D=(D_{jk})_{1\leq j,k\leq p}\in\mb{R}^{p^2\times p^2}$ with 
	$$D_{jk}=\frac{\Omega_x\mE[\psi(\varepsilon_{ij})\psi(\varepsilon_{ik})]\mE[X_iX_i^\top w^2(X_i)]\Omega_x^\top}{\mu_j\mu_k}\in\mb{R}^{p\times p}.$$ 
	Under Assumption (A4) and (A5), we have the following Gaussian Approximation result that
	$$\sup_{t\in\mb{R}}\bigg|\mP\bigg(|\sqrt n\Omega\nabla L_n(\beta^*)|_\infty\leq t\bigg)-\mP\bigg(\big|\sum_{i=1}^nz_i/\sqrt n\big|_\infty\leq t\bigg)\bigg|=o(1),$$
where $z_i=(z_{i1},\dots, z_{ip^2})^\top$ is a sequence of mean zero independent Gaussian vectors with each $\mE z_iz_i^\top=D$.
\end{thm}
\begin{rem}
	The above GA results allows the ultra-high dimensional regime, wehere $p$ grows as fast as $O(\me^{n^b})$ for some $0<b<1$. 
\end{rem}
Since the covariance matrix $D$ of the Gaussian analogue $z_i$ is not accessible from the observation $\{X_i\}$, we need to give a suitable estimation of $D$ before further performing multiplier bootstrap. The next theorem delivers a consistent estimator for our purpose.
\begin{thm}\label{cov estimation}
		\begin{equation}\label{eqhatD}
		\widehat D_{jk}=\frac{\widehat\Omega_x\Big(\frac1n\sum_{i=1}^n\psi(\widehat\varepsilon_{ij})\psi(\widehat\varepsilon_{ik})\Big)\Big(\frac1n\sum_{i=1}^nX_iX_i^\top w^2(X_i)\Big)\widehat\Omega_x^\top}{\widehat\mu_j\widehat\mu_k}\in\mb{R}^{p\times p},\end{equation}
		where $\widehat\Omega_x$ is the CLIME estimator of $\Omega_x$. Under assumptions (A1)--(A5) and additionally assume that $\|\Omega_x\|_1=O(\sqrt s)$ and that for all $1\leq k\leq p$, $\mu_k>C>0$ for some constant $C$, we have 
	with probability at least $1-12p^{-c}$, we have
	$$\|\widehat D-D\|_{\max}\lesssim s\gamma\tau^2T^2(\log p)^{3/2}n^{-1/2}+|\widehat\beta-\beta^*|_1.$$
\end{thm}
Indeed, under the scaling assumptions, $\|\widehat D-D\|_{\max}=o(1)$. With these preparatory results, we are ready to present the main theorem of this paper, which describes a procedure to find the critical value of $\sqrt n|\check\beta-\beta^*|_{\infty}$ using bootstrap.
\begin{thm}\label{thm2.10}
	Denote
	$$W=|\widehat D^{1/2}\eta|_\infty,$$
	where $\eta\sim N(0,I_{p^2})$ is independent of $(X_i)_{i=1}^n$ and $\widehat D$ is defined in \eqref{eqhatD}. Let the bootstrap critical value be given by $c(\alpha)=\inf\{t\in\mb{R}: \mP(W\leq t|\boldsymbol{X})\geq 1-\alpha\}$. Let assumptions (A1) --- (A5) and the assumptions in theorem \ref{thm2.1} hold. 
	Denote $v=c(s\gamma\tau^2T^2(\log p)^{3/2}/\sqrt n+|\widehat\beta-\beta^*|_1)$ for some constant $c$. Assume that $\pi(v)=Cv^{1/3}(1\lor\log(p/v))^{2/3}=o(1)$, then
	we have 
	$$\sup_{\alpha\in(0,1)}\bigg|\mP\Big(\sqrt n|\check\beta-\beta^*|_\infty>c(\alpha)\Big)-\alpha\bigg|=o(1).$$
\end{thm}
This result suggests a way to not only find the asymptotic distribution, but also to provide an accurate critical value $c(\alpha)$ using multiplier bootstrap. Under the null hypothesis $H_0$, we have $\sqrt n|\check\beta-\beta^0|_{\infty}=\sqrt n|\check\beta-\beta^*|$. This verifies the validity of having \eqref{eq2.9} as a test statistics for simultaneous inference. 

\section{Estimation}\label{estimation}
Many estimation tasks are needed as preparatory results for proving Theorem \ref{thm2.10}. For instance, Theorem \ref{thm2.10} requires an estimation of the theoretical covariance matrix $D$ of the Gaussian analogue $Z$, as stated in Theorem \ref{cov estimation}. Besides, the convergence of CLIME estimator (section 2.3) depends on the convergence of corresponding covariance matrix. Therefore, these problems requires us to develop a new estimation theory that delivers the convergence even in ultra-high dimensional regime.

The success of high-dimensional estimation relies heavily on the application of probability concentration inequality, among which Bernstein-type inequality is especially important. The celebrated Bernstein's inequality (\cite{bernstein1946theory}) provides an exponential concentration inequality for sums of independent random variables which are uniformly bounded. Later works relaxed the uniform boundedness condition and extended the validity of Bernstein inequality to independent random variables that have finite exponential moment; see for example, \cite{massart2007concentration} and \cite{wainwright2019high}.

Despite the extensive body of work on concentration inequalities for independent random variables, literature remains quiet when it comes to establishing exponential-type tail concentration results for random process. Some related existing work includes Bernstein inequality for sums of strong mixing processes (\cite{merlevede2009bernstein}), Bernstein inequality under functional dependence measures (\cite{zhang2019robust}), etc. In a more recent work, \cite{liu2021robust} established a sharp Bernstein inequality for VAR model using the definition of spectral decay index, which improved the current rate by a factor of $(\log n)^2$. In this paper, we will derive another Bernstein inequality for VAR model under slightly different condition from \cite{liu2021robust}. Before presenting the main results, recall the definition of $\tau$ in definition \ref{defn2.1}.

\begin{lem} \label{lem3.1}
Let $\{X_i\}_{i=0}^n$ be generated by a VAR(1) model. Suppose $G:\mb{R}^p\to\mb{R}$ satisfies that 
\begin{equation}\label{eq3.1}
|G(X)-G(Y)|\leq|X-Y|_\infty,
\end{equation}
and that $|G(x)|\leq B$ for all $x\in\mb{R}$. Assume that $\mE[|\varepsilon_{ij}|^2]\leq \sigma^2$ for all $j=1,\dots,p$. Then there exists some constants $C_1,C_2,C_3,C_4>0$ only depending on $\rho$ and $\sigma$, such that
\begin{align*}\mP\bigg(\Big|\frac1n\sum_{i=1}^n G(X_{i-1})-\mE[&G(X_{i-1})]\Big|\geq x\bigg)\leq2\exp\bigg\{-\frac{nx^2}{C_3n^{-1}\gamma^2\tau^3+C_4\tau Bx}\bigg\}\\
&+2\exp\bigg\{-\frac{nx^2}{(1+C_1B^{-2})\gamma^2\tau^4B^2(\log p)^2(n^{-1}\tau\log p+1)+C_2\tau^2 B(\log p)x}\bigg\}.
\end{align*}
Specifically, under assumption (A2), we see that $\tau(\log p)/n\to0$. So for sufficiently large $B>0$, we have 
\begin{align}\label{eq3.2}\mP\bigg(\Big|\frac1n\sum_{i=1}^n G(X_{i-1})-\mE[G(X_{i-1})]\Big|\geq x\bigg)
\leq 4\exp\bigg\{-\frac{nx^2}{C_1'\gamma^2\tau^4B^2(\log p)^2+C_2'\tau^2 B(\log p)x}\bigg\},
\end{align}
for some positive constants $C_1',C_2'$ depending only on $\rho$ and $\sigma$.
\end{lem}
\begin{rem}
	Note that the Lipschitz condition \eqref{eq3.1} is slightly different from that in \cite{liu2021robust}, where instead, they assumed that 
	\begin{equation}\label{eq3.3}|G(x)-G(y)|\leq g^\top|x-y|,\end{equation}
	for some vector $g\in\mb{R}^p.$ Since condition \eqref{eq3.1} is weaker than \eqref{eq3.3}, the additional $(\log p)$ appears in the denominator of right-hand side in \eqref{eq3.2}. For more detailed comparison of different versions of Bernstein inequalities, we refer readers to \cite{liu2021robust} and the references therein.
\end{rem}
With a minor modification of the proof of Lemma \ref{lem3.1}, we have the following version of Bernstein inequality which includes a bounded function of the latest innovation $\varepsilon_{i}$ as a multiple.

\begin{cor} \label{cor3.2}
Let $\{X_i\}_{i=0}^n$ be generated by a VAR(1) model. Suppose  $|h(x)|\leq 1$ and $G:\mb{R}^p\to\mb{R}$ satisfies that $$|G(X)-G(Y)|\leq|X-Y|_\infty,$$
and that $|G(x)|\leq B$ for all $x\in\mb{R}$. Assume that $\mE[|\varepsilon_{ij}|^2]\leq \sigma^2$ for all $j=1,\dots,p$. Then there exists some constants $C_1,C_2,C_3,C_4>0$ only depending on $\rho$ and $\sigma$, such that
\begin{align*}\mP\bigg(\Big|\frac1n\sum_{i=1}^n h(\varepsilon_i)G(X_{i-1})&-\mE[h(\varepsilon_i)G(X_{i-1})]\Big|\geq x\bigg)\leq2\exp\bigg\{-\frac{nx^2}{C_3n^{-1}\gamma^2\tau^3+C_4\tau Bx}\bigg\}\\
&+2\exp\bigg\{-\frac{nx^2}{(1+C_1B^{-2})\gamma^2\tau^4B^2(\log p)^2(n^{-1}\tau\log p+1)+C_2\tau^2 B(\log p)x}\bigg\}.
\end{align*}
Specifically, under assumption (A2), we see that $\tau(\log p)/n\to0$. So for sufficiently large $B>0$, we have 
\begin{align*}\mP\bigg(\Big|\frac1n\sum_{i=1}^n h(\varepsilon_i)G(X_{i-1})&-\mE[h(\varepsilon_i)G(X_{i-1})]\Big|\geq x\bigg)
\leq 4\exp\bigg\{-\frac{nx^2}{C_1'\gamma^2\tau^4B^2(\log p)^2+C_2'\tau^2 B(\log p)x}\bigg\},
\end{align*}
for some positive constants $C_1',C_2'$ depending only on $\rho$ and $\sigma$.
\end{cor}
\begin{rem}
	Since the additional term $h(\varepsilon_i)$ is independent of $G(X_{i-1})$, the proof of Lemma \ref{lem3.1} directly applies without any extra technical difficulty.
\end{rem}

Equipped with our new Bernstein inequalities, several estimation results follow immediately. The next theorem regarding the estimation of $\Sigma_x$ is essential when we prove the convergence rate of CLIME estimator in section 2.3. 
\begin{thm}[Estimation of $\Sigma_x$]\label{thm3.3}
	Let $\widehat\Sigma_x=n^{-1}\sum_{i=1}^n X_{i-1}X_{i-1}^\top w(X_{i-1})$ and $\Sigma_x=\mE[X_iX_i^\top w(X_i)]$. Then with probability at least $1-4p^{-c_0}$ for some constant $c_0>0$, it holds that
	$$\|\widehat\Sigma_x-\Sigma_x\|_{\max}\lesssim\gamma\tau^2T^2 n^{-1/2}(\log p)^{3/2}.$$
\end{thm}
We see that the convergence rate of CLIME estimator in Theorem \ref{thm4.1} essentially inherits from the convergence rate in Theorem \ref{thm3.3}, with an additional term $\|\Omega_x\|_1$. The following theorem plays an important role in verifying that the $\Delta$ defined in \eqref{eq2.8} is indeed negligible.

\begin{thm}[Estimation of $\Sigma$ by $\nabla^2 L_n(\beta^*)$]\label{thm3.4}
	Assume that $\mE[\varepsilon_{ik}^2]\leq\sigma^2$ for all $1\leq k\leq p$. Then for some constant $c_1>0$, with probability at least $1-4p^{-c_1}$, it holds that
	$$\|\nabla^2 L_n(\beta^*)-\Sigma\|_{\max}\lesssim \gamma\tau^2T^2 n^{-1/2}(\log p)^{3/2}.$$
\end{thm}
While the last two theorems make use of Lemma \ref{lem3.1} in this paper, the next estimation for $\mu$ directly applies the concentration inequality in \cite{liu2021robust} thanks to the stronger assumption that $\widehat\mu$ satisfies.
\begin{lem}\label{lem1.1}
	Suppose that $\beta_k^*$ lies in a bounded $\ell_1$ normed ball for all $1\leq k\leq p$ and that $\mE[X_{ij}^2]\leq C$ for some constant $C>0$ and for all $1\leq j\leq p$. Then we have
	$$\mP\bigg(|\widehat\mu-\mu|_{\infty}\geq\gamma\tau^2\sqrt{\frac{\log p}{n}}+|\widehat\beta-\beta^*|_1\bigg)\leq 2p^{-c},$$
	for some positive constant $c$.
\end{lem}

\section{Gaussian Approximation}\label{GA}
Conducting simultaneous inference for high-dimensional data is always considered to be a hard task, since central limit theorem fails when the dimension of random vectors can grow as a function of the number of observation $n$, or even exceeds $n$. As an alternative to central limit theorem, \cite{chernozhukov2013gaussian} proposed Gaussian approximation theorem, which states that under certain conditions, the distribution of the maximum of a sum of independent high-dimensional random vectors can be approximated by that of the maximum of a sum of the Gaussian random vectors with the same covariance matrices as the original vectors. Their Gaussian approximation results allow the ultra-high dimensional cases, where the dimension $p$ grows exponentially in $n$. In the meantime, they also proved that Gaussian multiplier bootstrap method yields a high quality approximation of the distribution of the original maximum and showcased a wide range of application, such as high-dimensional estimation, multiple hypothesis testing, and adaptive specification testing. It is worth noticing that the results from \cite{chernozhukov2013gaussian} are only applicable when the sequence of random vectors is independent. 

\cite{zhang2017gaussian} generalized Gaussian approximation results to general high-dimensional stationary time series, using the framework of functional dependence measure (\cite{wu2005nonlinear}). We specifically mention that a direct application of Gaussian approximation from \cite{zhang2017gaussian} cannot deliver a desired conclusion in ultra-high dimensional regime, due to coarser capture of dependence measure for VAR model. In what follows, we will use refined argument to establish a new Gaussian approximation result for VAR model.

By Theorem \ref{thm2.1}, $\sqrt n|\check\beta-\beta^*|_{\infty}$ can be approximated by $\sqrt n|\Omega\nabla L_n(\beta^*)|_{\infty}$. Hence, we shall build a GA result for $\sqrt n\Omega\nabla L_n(\beta^*)$. Observe that $\sqrt n\Omega\nabla L_n(\beta^*)\in\mb{R}^{p^2}$ can be written as
\begin{equation*}
	\bigg(\frac1{\sqrt n}\sum_{i=1}^n\frac{\Omega_x}{\mu_1}\psi_\alpha(\varepsilon_{i1})X_{i-1}^\top w(X_{i-1}),\dots,\frac1{\sqrt n}\sum_{i=1}^n\frac{\Omega_x}{\mu_p}\psi_\alpha(\varepsilon_{ip})X_{i-1}^\top w(X_{i-1}),\bigg)^\top,
\end{equation*}
so it's sufficient to establish GA result for one sub-vector 
$$\frac1{\sqrt n}\sum_{i=1}^n\frac{\Omega_x}{\mu_k}\psi_\alpha(\varepsilon_{ik})X_{i-1}^\top w(X_{i-1}), \quad k=1,\dots, p.$$

Fix $1\leq k\leq p$ and denote $\Theta_k=\Omega_x\mu_k^{-1}$.
Let $X_{i,m}=\sum_{l=0}^mA^l\varepsilon_{i-l}$ be the $m$-approximation of $X_i$ with $m$ to be determined later. Let $Y_i=\psi_\alpha(\varepsilon_{ik})\Theta_kX_{i-1}w(X_{i-1})$ be the quantity that we will establish Gaussian approximation for and denote $T_Y=\sum_{i=1}^nY_i$. Analogously, let $Y_{i,m}=\psi_\alpha(\varepsilon_{ik})\Theta_kX_{i-1,m}w(X_{i-1,m})$ be the $m$-approximation of $Y_i$ and write $T_{Y,m}=\sum_{i=1}^nY_{i,m}$. For simplicity, assume $n=(m+M)w$, where $M\to\infty,$ $m\to\infty$ , $w\to\infty$ and $m/M\to 0$. Divide the interval $[1,n]$ into alternating large blocks $L_b=[(b-1)(M+m)+1,bM+(b-1)m]$ with $M$ points and small blocks $S_b=[bM+(b-1)m+1,b(M+m)]$ with $m$ points, for $1\leq b\leq w$. Denote
\begin{align*}
	\xi_b&=\sum_{i\in L_b}Y_{i,m}/\sqrt M,\quad T_{Y,S}=\sum_{b=1}^w\sum_{i\in S_b}Y_{i,m},\quad T_{Y,L}=\sum_{b=1}^w\sum_{i\in L_b}Y_{i,m},\\
	Z&=\frac1{\sqrt n}\sum_{i=1}^n U_i,\quad\text{where }U_i\sim N(0,\mu_k^{-2}\mE[\psi^2(\varepsilon_{ik})]\Omega_x\mE[X_iX_i^\top w^2(X_i)]\Omega_x^\top)
\end{align*}
Note that the $Y_{i,m}$ from different large blocks $L_b$ are independent, i.e. $\sum_{i\in L_b}Y_{i,m}$ is independent in $b=1\dots,w$.
The main result of this section is presented as follow.
\begin{thm}\label{thm5.4}
	Suppose $\mE[\varepsilon_{ik}^2]\leq\sigma^2$ for all $1\leq k\leq p$ and the odd function $\psi(\cdot)$ satisfies that $|\psi(\cdot)|\leq C$ and $|\psi'(\cdot)|\leq C$.  Suppose the scaling condition holds that $sT^2(\log(pn))^7/n\leq c_1n^{-c_2}$. Then for any $\eta>0$, the Gaussian Approximation holds that 
	\begin{align}\label{eq5.7}
		\mathcal{H}:&=\sup_{t\in\mb{R}}\bigg|\mP\big(|T_{Y}/\sqrt n|_\infty\leq t\big)-\mP\big(|Z|_\infty\leq t\big)\bigg|\notag\\
		&\lesssim f_1(\eta/2,m)+f_2(\eta/2,m)+\eta\sqrt{\log p}+\eta\sqrt{\log(1/\eta)}+cn^{-c'},
	\end{align}
	for some $c,c'>0$.
\end{thm}
This theorem gives an upper bound on the supremum of the difference between the distribution of the maximum of sum of $Y_i$ and that of the maximum of sum of Gaussian vectors $U_i$ with the same covariance. Now, we present the outline of the proof of the previous theorem, while we leave the complete proof in the appendix. 

First, we show that the sum of $Y_{i,m}$ in the small blocks are negligible, so $T_{Y,m}\approx T_{Y,L}$. Next, we prove that the sum of $Y_i$ can be approximated by its $m$-approximation, that is, $T_Y\approx T_{Y,m}\approx T_{Y,L}$. Since $T_{Y,L}$ is a sum of independent random vector $\{\sum_{i\in L_b}Y_{i,m}\}_{b=1}^w$, the GA theorem from \cite{chernozhukov2013gaussian} can be applied.

\section{Numerical Experiments}\label{num}
In this section, we evaluate the performance of the proposed bootstrap-assist procedure in simultaneous inference. We consider the model \eqref{var1}, where $\varepsilon_{ij}$'s are i.i.d.~Student's $t$-distributions with $df=5$ or $10$. Let $s =\lfloor \log p \rfloor $. We pick $n=30$ and $p=10$ in the numerical setup. For the true transition matrix $A=(a_{ij})$, we consider the following designs.
\begin{itemize}
	\item [(1)] Banded: $A = (\lambda^{|i-j|}{\bf 1}\{|i-j|\leq s\})$ and $\lambda = 0.5$.
	\item [(2)] Block diagonal: $A = \text{diag}\{A_i\}$, where each $A_i \in \mathbb{R}^{s\times s}$ has $\lambda_i$ on the diagonal and $\lambda_i^2$ on the superdiagonal with $\lambda_i \sim Unif(-0.8, 0.8).$
\end{itemize}
The design in (1) is further scaled by $2\rho(A)$ to ensure that $\rho(A)<1$. Hence sparse symmetric matrices are generated in (1) and sparse asymmetric matrices are constructed in (2).  We draw the qq-plots of the data quantile of $\sqrt n|\check\beta-\beta^*|_{\infty}$ versus the data quantile of $W$ defined in Theorem \ref{thm2.10} from $m=100$ duplicates. The qq-plots are shown in figure \ref{banded} and figure \ref{block} for banded and block diagonal designs respectively.
\begin{figure}[H]
\begin{minipage}[t]{\linewidth}
\centering
\includegraphics[width=12cm]{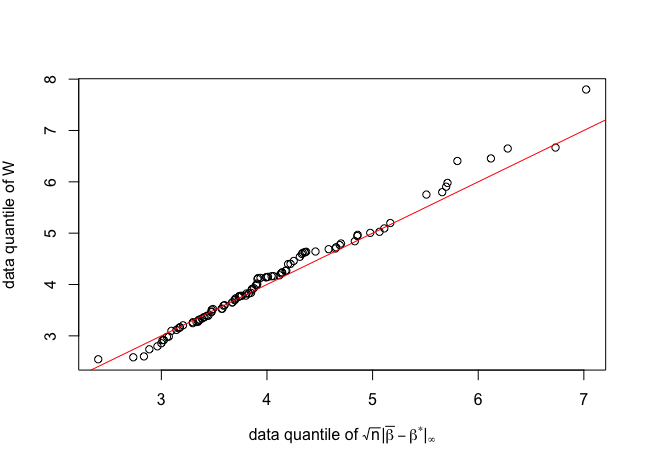}
\caption{The qq-plot of banded design.}
\label{banded}
\end{minipage}%
\\
\begin{minipage}[t]{\linewidth}
\centering
\includegraphics[width=12cm]{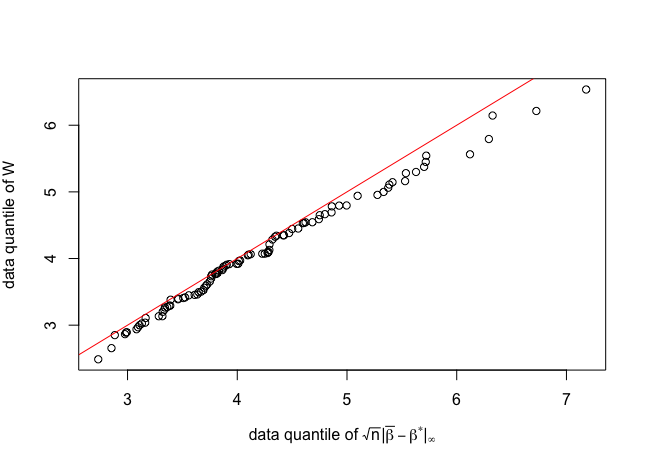}
\caption{The qq-plot of block diagonal design.}
\label{block}
\end{minipage}
\end{figure}



\appendix
\section{Proofs of Results in Section \ref{main}}
Before proceeding with the proofs, we state a helpful lemma that is repeatedly used throughout the paper and present its proof. This simply lemma is an application of the triangle inequality to the product of two matrices.
\begin{lem}\label{lem1.2}
	Let $A, B$ and $\widehat A, \widehat B$ be $p\times p$ symmetric matrices and $\|A-\widehat A\|_{1}=o(1)$. Suppose $\|A\|_1=O(1)$ and $\|B\|_1=O(1)$. Then $\|AB-\widehat A\widehat B\|_{\max}\lesssim \|A-\widehat A\|_{\max}+\|B-\widehat B\|_{\max}$.
\end{lem}
\begin{proof}[Proof of Lemma \ref{lem1.2}]
	Since $\|A\|_1=O(1)$ and $\|A-\widehat A\|_1=o(1)$, $\|\widehat A\|_1\leq \|A-\widehat A\|_1 + \|A\|_1=O(1)$. Hence, 
	by triangular inequality,
	\begin{align*}
		\|AB-\widehat A\widehat B\|_{\max}&\leq\|(A-\widehat A)B\|_{\max} + \|\widehat A(B-\widehat B)\|_{\max}\\
		&\leq \|B\|_1\|A-\widehat A\|_{\max} + \|\widehat A\|_1\|B-\widehat B\|_{\max}\\
		&\lesssim \|A-\widehat A\|_{\max} + \|B-\widehat B\|_{\max}
	\end{align*}
\end{proof}

\begin{proof}[Proof of Theorem \ref{thm4.1}]
	By Theorem \ref{thm3.3}, with probability at least $1-4p^{-c_0}$,
	$$\|\widehat\Sigma_x-\Sigma_x\|_{\max}\leq\lambda_n.$$
	By Theorem 6 of \cite{cai2011constrained}, we have the desired result.
\end{proof}

\begin{proof}[Proof of Theorem \ref{thm4.2}]
	Recall that 
	\begin{align}
		\Omega&=\Omega_x\otimes\text{diag}(\mu^{-1})=\begin{bmatrix}
			\mu_1^{-1}\Omega_x & 0 & 0 & \dots & 0\\
			0 & \mu_2^{-1}\Omega_x & 0 & \dots & 0\\
			\vdots & \vdots & \ddots & \dots & 0\\
			0 & 0& 0 & 0 & \mu_p^{-1}\Omega_x
		\end{bmatrix}
	\end{align}
	and
	\begin{align}
		\widehat\Omega&=\widehat\Omega_x\otimes\text{diag}(\widehat\mu^{-1})=\begin{bmatrix}
			\widehat\mu_1^{-1}\widehat\Omega_x & 0 & 0 & \dots & 0\\
			0 & \widehat\mu_2^{-1}\widehat\Omega_x & 0 & \dots & 0\\
			\vdots & \vdots & \ddots & \dots & 0\\
			0 & 0& 0 & 0 & \widehat\mu_p^{-1}\widehat\Omega_x
		\end{bmatrix}
	\end{align}
	For $1\leq k\leq p$, consider
	\begin{align*}
\|\widehat\Omega_x\widehat\mu_k^{-1}-\Omega_x\mu_k^{-1}\|_{\max}&\leq \|\widehat\Omega_x-\Omega_x\|_{\max}|\widehat\mu_k^{-1}|+\|\Omega_x\|_{\max}|\widehat\mu_k^{-1}-\mu_k^{-1}|\\
&\lesssim \|\Omega_x\|_1\lambda_n + \|\Omega_x\|_{\max}\frac{|\mu_k-\widehat\mu_k|}{\mu_k\widehat\mu_k}\\
&\lesssim \|\Omega_x\|_{1}\lambda_n,
	\end{align*}
	with probability no less than $1-6p^{-c'}$ by theorem \ref{thm4.1} and lemma \ref{thm3.4}. Taking a union bound for all $k$ yields
	$$\|\Omega-\widehat\Omega\|_{\max}=\max_{1\leq k\leq p}\|\mu_k^{-1}\Omega_x-\widehat\mu_k^{-1}\widehat\Omega_x\|_{\max}\lesssim\|\Omega_x\|_1\lambda_n,$$
	with probability at least $1-6p^{-(c'-1)}$. Replacing $\max$-norm by $L_1$-norm delivers
	$$\|\Omega-\widehat\Omega\|_{1}\lesssim\|\Omega_x\|_1s\lambda_n.$$
\end{proof}

The next Lemma provides a high probability bound on $|\nabla L_n(\beta^*)|_{\infty}$, which will be used in the proof of Theorem \ref{thm2.1}.

\begin{lem}\label{lem1.3}
	Suppose that $\mE[\varepsilon_{ij}^2]\leq C$ for all $1\leq j\leq p$. Then it holds that
	$$\mP(|\nabla L_n(\beta^*)|_\infty\gtrsim\gamma\tau^2 T(\log p)^{3/2}/\sqrt n)\leq4p^{-c},$$
	for some constant $c>0$.
\end{lem}
\begin{proof}[Proof of Lemma \ref{lem1.3}]
	We shall apply Corollary \ref{cor3.2}. Consider the first coordinate $\nabla L_{n1}(\beta^*)$ of $\nabla L_n(\beta^*)$. In Corollary \ref{cor3.2}, let $h(\varepsilon_i)=\psi(\varepsilon_{i1})$ and $G(X_i)=X_{i1}w(X_i).$ Observe that $\mE[\nabla L_n(\beta^*)]=0$. By Corollary \ref{cor3.2}, 
	\begin{align*}
		\mP(|\nabla L_{n1}(\beta^*)|\geq x)&=\mP(|\nabla L_{n1}(\beta^*)-\mE[\nabla L_{n1}(\beta^*)]|_\infty\geq x)\\
		&=\mP\bigg(\Big|\frac1n\sum_{i=1}^nh(\varepsilon_i)G(X_{i-1})-\mE[h(\varepsilon_i)G(X_{i-1})]\Big|\geq x\bigg)\\
		&\leq4\exp\bigg\{-\frac{nx^2}{C_1\gamma^2\tau^4T^2(\log p)^2+C_2\tau^2 T(\log p)x}\bigg\}
	\end{align*}
	Choose $x=c'\gamma\tau^2 T(\log p)^{3/2}/\sqrt n$ and we get
	$$\mP(|\nabla L_{n1}(\beta^*)|\geq c'\gamma\tau^2 T(\log p)^{3/2}/\sqrt n)\leq4p^{-c},$$
	for some constant $c>0$. Take sufficiently large $c'$ such that $c>1$, so by a union bound we obtain
	$$\mP(|\nabla L_{n}(\beta^*)|_\infty\geq c'\gamma\tau^2 T(\log p)^{3/2}/\sqrt n)\leq4p^{-c''},$$
	where $c''=c-1>0$.
\end{proof}

\begin{proof}[Proof of Theorem \ref{thm2.1}]
 	By Taylor expansion, we write 
 	\begin{align*}
	\sqrt n(\check\beta-\beta^*)&=\sqrt n(\widehat\beta-\beta^*)+\sqrt n\,\widehat\Omega\,\nabla L_n(\beta^*)-\sqrt n\,\widehat\Omega(\nabla L_n(\widehat\beta)-\nabla L_n(\beta^*))\notag\\
	&=\sqrt n\,\widehat\Omega\,\nabla L_n(\beta^*)+\sqrt n\left[(\widehat\beta-\beta^*)-\widehat\Omega\,\nabla^2 L_n(\beta^*)(\widehat\beta-\beta^*)+ R\right]\notag\\
	&=\underbrace{\sqrt n \,\widehat\Omega\,\nabla L_n(\beta^*)}_A+\underbrace{\sqrt n\left[\Big(I_{p^2}-\widehat\Omega\,\nabla^2L_n(\beta^*)(\widehat\beta-\beta^*)\right]}_\Delta+\sqrt nR,
\end{align*}
where the remainder
$$R=\frac1{2n}\sum_{i=1}^n\Big(\psi''(z_{i1})\big(X^\top_{i-1}(\widehat\beta_1-\beta^*_1)\big)^2X_{i-1}^\top w(X_{i-1}),\dots,\psi''(z_{ip})\big(X^\top_{i-1}(\widehat\beta_p-\beta^*_p)\big)^2X_{i-1}^\top w(X_{i-1})\Big)^\top,$$
where $z_{ik}=X_i-X_{i-1}^\top\widetilde\beta$ for some $\widetilde\beta$ lying between $\beta^*$ and $\widehat\beta$.
Now we analyze the above terms $R,\Delta,A$ respectively. First we see that  $\sqrt n|R|_\infty=O_{\mP}(\sqrt nT^3|\widehat\beta-\beta^*|_1^2)=o(1)$
by assumption (A1). To analyze $\Delta$, denote $H=\nabla^2L_n(\beta^*)$. Then we write
$$\Delta=\sqrt n\Big(I_{p^2}-\widehat\Omega\,H\Big)(\widehat\beta-\beta^*)=\sqrt n\Big(\Omega\Sigma-\widehat\Omega H\Big)(\widehat\beta-\beta^*).$$
Thus, by theorem \ref{thm3.4} and theorem \ref{thm4.2}, with probability tending to 1,
\begin{align*}
	|\Delta|_\infty&\leq\sqrt n\|\Omega\Sigma-\widehat\Omega H\|_{\max}|\widehat\beta-\beta^*|_{1}\leq \sqrt n\Big(\|\Omega-\widehat\Omega\|_{1}\|\Sigma\|_{\max}+\|H-\Sigma\|_{\max}\|\widehat\Omega\|_1\Big)|\widehat\beta-\beta^*|_1\\
	&\lesssim\sqrt n\|\Omega_x\|_{1}\lambda_n|\widehat\beta-\beta^*|_1
	\asymp s\gamma\tau^2T^2(\log p)^{3/2}|\widehat\beta-\beta^*|_1=o(1)
\end{align*}
by assumption (A3).
Finally, by Lemma \ref{lem1.3} and Theorem \ref{thm4.2}, with probability tending to 1, it holds that
\begin{align*}|A-\sqrt n\Omega\nabla L_n(\beta^*)|_\infty&\leq\|\widehat\Omega-\Omega\|_1|\sqrt n \nabla L_n(\beta^*)|_\infty\leq \|\Omega_x\|_1^2s\gamma^2\tau^4T^4(\log p)^3/\sqrt n\\
&\asymp s^2\gamma^2\tau^4T^4(\log p)^3/\sqrt n.\end{align*}	
Therefore, 
$$|\sqrt n(\check\beta-\beta^*)-\sqrt n\Omega\nabla L_n(\beta^*)|_{\infty}\leq|\sqrt n(\check\beta-\beta^*)-A|_{\infty}+|A-\sqrt n\Omega\nabla L_n(\beta^*)|_{\infty}\leq\zeta_1,$$
where 
$$\zeta_1=s\gamma\tau^2T^2(\log p)^{3/2}|\widehat\beta-\beta^*|_1+\sqrt nT^3|\widehat\beta-\beta^*|_1^2+s^2\gamma^2\tau^4T^4(\log p)^3/\sqrt n.$$
 \end{proof}

\begin{proof}[Proof of Theorem \ref{thm2.2}]
	The proof of Theorem \ref{thm5.4} can be easily generalized to $p^2$ dimensional space, thus it still holds for $|\sqrt n\Omega\nabla L_n(\beta^*)|_\infty$. By Theorem \ref{thm5.4}, we have for any $\eta>0$,
		\begin{align}\label{eq2.1}
		&\sup_{t\in\mb{R}}\bigg|\mP\big(|\sqrt n\Omega\nabla L_n(\beta^*)|_\infty\leq t\big)-\mP\big(|Z|_\infty\leq t\big)\bigg|\notag\\
		\lesssim{}& f_1(\eta/2,m)+f_2(\eta/2,m)+\eta\sqrt{\log p}+\eta\sqrt{\log(1/\eta)}+cn^{-c'},
	\end{align}
	where 
	\begin{equation}\label{eq2.2}f_1(x, m)=\frac{c_1sp^3\gamma^2\rho^{m/\tau}}{x^2}\quad\text{and}\quad f_2(x)=2p\exp\bigg\{-\frac{nx^2}{2\sqrt sTM\sqrt nx+4mwsT^2\sigma^2}\bigg\}.
	\end{equation}	
Now choose $\eta\asymp(\log p)T\sqrt{s\tau}\gamma/n^{1/4}=o(1),\omega\asymp n^{1/2}, M\asymp n^{1/2}, m=c\tau\log p$ in \eqref{eq2.1} for some constant $c>0$. For sufficiently large $c$, basic algebra shows that
\begin{align}
	f_1(\eta/2, m)\lesssim\frac{s\gamma^2}{p^{c-3}\eta^2}\asymp\frac{n^{1/2}}{p^{c-3}T^2\tau(\log p)^2}=o(1),
\end{align}
since the order of $p^{c-3}$ dominates the order of $n^{1/2}$. Moreover, 
\begin{align}
	f_2(\eta/2,m)\leq 2p\exp\bigg\{-\frac{c_1\gamma^2\log p}{c_2\gamma+c_3}\bigg\}\leq 2p\exp\{-c_4\log p\}=o(1),
\end{align}
	by a proper choice of constant $c_1,c_2,c_3$.
Also, by assumption (A5), $\eta\sqrt{\log p}=o(1)$ and $$\eta\sqrt{\log(1/\eta)}\lesssim\frac{T\sqrt{s\tau}\gamma\log p}{n^{1/4}}\sqrt{\log n}=o(1).$$
	Thus the proof is completed.
\end{proof}

\begin{proof}[Proof of Theorem \ref{cov estimation}]
	First, we collect several useful results.
	\begin{itemize}
		\item[(i)] With probability at least $1-4p^{-c_1}$, $\|\Omega_x-\widehat\Omega_x\|_1\leq \|\Omega_x\|_1^2s\gamma\tau^2T^2(\log p)^{3/2}n^{-1/2}$ and $\|\Omega_x-\widehat\Omega_x\|_{\max}\leq \|\Omega_x\|_1^2\gamma\tau^2T^2(\log p)^{3/2}n^{-1/2}$ by Theorem \ref{thm4.1}. Therefore, $\|\Omega_x-\widehat\Omega_x\|_1=o(1)$ and $\|\Omega_x-\widehat\Omega_x\|_{\max}=o(1)$ by assumption (A2).
		\item[(ii)] With probability at least $1-2p^{-c_2}$, $|\mu-\widehat\mu|_\infty\lesssim \gamma\tau^2\sqrt{\frac{\log p}n}+|\widehat\beta-\beta^*|_1=o(1)$ Lemma \ref{lem1.1} and the order comes from assumptions (A1) and (A2).
		\item[(iii)] Similar to the proof of Lemma \ref{lem1.1}, we have with probability at least $1-2p^{-c_3}$, $$\Big|\frac1n\sum_{i=1}^n\psi(\widehat\varepsilon_{ij})\psi(\widehat\varepsilon_{ik})-\mE[\psi(\varepsilon_{ij})\psi(\varepsilon_{ik})]\Big|_\infty\lesssim \gamma\tau^2\sqrt{\frac{\log p}n}+|\widehat\beta-\beta^*|_1=o(1)$$
		\item[(iv)] Similar to the proof of Lemma \ref{thm3.3}, we have with probability at least $1-4p^{-c_4}$, $$\Big\|\frac1n\sum_{i=1}^nX_iX_i^\top w^2(X_i)-\mE[X_iX_i^\top w^2(X_i)]\Big\|_{\max}\lesssim\gamma\tau^2T^2(\log p)^{3/2}n^{-1/2}=o(1).$$
	\end{itemize}
	
	Repeatedly using Lemma \ref{lem1.2}, we get 
	\begin{align*}
	\|\widehat D-D\|_{\max}&\lesssim \max_{1\leq j,k\leq p}\Big|\frac1{\mu_j\mu_k}-\frac1{\widehat \mu_j\widehat \mu_k}\Big|+\Big|\frac1n\sum_{i=1}^n\psi(\widehat\varepsilon_{ij})\psi(\widehat\varepsilon_{ik})-\mE[\psi(\varepsilon_{ij})\psi(\varepsilon_{ik})]\Big|_\infty\\
	&+2\|\Omega_x-\widehat\Omega_x\|_{\max} + \Big\|\frac1n\sum_{i=1}^nX_iX_i^\top w^2(X_i)-\mE[X_iX_i^\top w^2(X_i)]\Big\|_{\max}\\
	&\lesssim \gamma\tau^2\sqrt{\frac{\log p}n}+|\widehat\beta-\beta^*|_1+\gamma\tau^2T^2(\log p)^{3/2}n^{-1/2}\\
	&\lesssim \gamma\tau^2T^2(\log p)^{3/2}n^{-1/2}+|\widehat\beta-\beta^*|_1
	\end{align*}
	with probability at least $1-12p^{-c}$, where $c=\min_{1\leq i\leq4} c_i$.
\end{proof}

\begin{proof}[Proof of Theorem \ref{thm2.10}]
	By theorem \ref{thm2.1}, we see that
		$$\mP\bigg(|\sqrt n(\check\beta-\beta^*)-\sqrt n\Omega\nabla L_n(\beta^*)|_\infty>\zeta_1\bigg)<\zeta_2,$$
	where $\zeta_1\sqrt{1\lor \log(p/\zeta_1)}=o(1)$ and $\zeta_2=o(1)$.
	Define $\pi(v)=Cv^{1/3}(1\lor\log(p/v))^{2/3}$ with $C_2>0$ and 
	$$\Gamma=\|\widehat D-D\|_{\max}.$$
	Let $c_z(\alpha)=\inf\{t\in\mb{R}: \mP(|\sum_{i=1}^nz_i/\sqrt n|_\infty\leq t)\geq 1-\alpha\}$, where the sequence $\{z_i\}$ is defined in theorem \ref{thm2.2}. From the proof of Lemma 3.2 in \cite{chernozhukov2013gaussian}, we have
	\begin{align}
		\mP\Big(c(\alpha)\leq c_z(\alpha+\pi(v))\Big)&\geq1-\mP(\Gamma>v)\label{eq3.6}\\
		\mP\Big(c_z(\alpha)\leq c(\alpha+\pi(v))\Big)&\geq1-\mP(\Gamma>v)\label{eq3.7}
	\end{align}
	Therefore, by theorem \ref{thm2.2}, \eqref{eq3.6} and \eqref{eq3.7}, we have for every $v>0$,
	\begin{align*}
		\sup_{\alpha\in(0,1)}\bigg|\mP\Big(\sqrt n\Omega\nabla L_n(\beta^*)>c(\alpha)\Big)-\alpha\bigg|&\lesssim \sup_{\alpha\in(0,1)}\bigg|\mP\Big(|\sum_{i=1}^nz_i/\sqrt n|_\infty>c(\alpha)\Big)-\alpha\bigg|+o(1)\\
		&\lesssim\pi(v)+\mP(\Gamma>v)+o(1)
	\end{align*}
	Furthermore, following the same spirit as the proof of Theorem 3.2 in \cite{chernozhukov2013gaussian}, we see that
	$$\sup_{\alpha\in(0,1)}\bigg|\mP\Big(\sqrt n|\check\beta-\beta^*|_\infty>c(\alpha)\Big)-\alpha\bigg|\lesssim\pi(v)+\mP(\Gamma>v)+\zeta_1\sqrt{1\lor\log(p/\zeta_1)}+\zeta_2+o(1).$$
	Now that $\zeta_1\sqrt{1\lor\log(p/\zeta_1)}=o(1)$ and $\zeta_2=o(1)$ from Theorem \ref{thm2.1}, we only need to choose $v>0$, such that $\pi(v)=o(1)$ and $\mP(\Gamma>v)=o(1)$. Let $v\asymp s\gamma\tau^2T^2(\log p)^{3/2}n^{-1/2}+|\widehat\beta-\beta^*|_1$. Then we see that the conditions that $\mP(\Gamma>v)=o(1)$ and $\pi(v)=o(1)$ are satisfied by Theorem \ref{cov estimation} and the scaling hypothesis.
\end{proof}

\section{Proofs of Results in Section \ref{estimation}}
\begin{proof}[Proof of Lemma \ref{lem3.1}]
	Define the filtration $\{\mathcal{F}_i\}$ with $\mathcal{F}_i=\sigma(\varepsilon_i,\varepsilon_{i-1},\dots)$, and let $P_j(\cdot)=\mE(\cdot|\mathcal{F}_j)-\mE(\cdot|\mathcal{F}_{j-1})$ be a projection. Conventionally it follows that $P_j(G(X_i))=0$ for $j\geq i+1$. We can write
$$\sum_{i=1}^nG(X_i)-\mE G(X_i)=\sum_{j=-\infty}^n\left(\sum_{i=1}^nP_j(G(X_i))\right)=:\sum_{j=-\infty}^nL_j,$$
where $L_j=\sum_{i=1}^nP_j(G(X_i))$. By the Markov inequality, for $\lambda>0$, we have 
\begin{align}\label{P}
	\mP\bigg(\sum_{i=1}^nG(X_i)-\mE G(X_i)\geq 2x\bigg)
	&\leq \mP\bigg(\sum_{j=-\infty}^{-s}L_j\geq x\bigg)+\mP\bigg(\sum_{j=-s+1}^nL_j\geq x\bigg)\notag\\
	&\leq \me^{-\lambda x}\mE\bigg[\exp\bigg\{\lambda\sum_{j=-\infty}^{-s}L_j\bigg\}\bigg]+\me^{-\lambda x}\mE\bigg[\exp\bigg\{\lambda\sum_{j=-s+1}^{n}L_j\bigg\}\bigg],
\end{align}
for some $s>0$ to be determined later.
We shall bound the right-hand side of \eqref{P} with a suitable choice of $\lambda>0$. Observing that $\{L_j\}_{j\leq n}$ is a sequence of martingale differences with respect to $\{\mathcal{F}_j\}$, we then seek an upper bound on $\mE[\me^{\lambda L_j}\bigr|\mF_{j-1}]$. 
It follows that
\begin{align}\label{lipscond}
	|L_j|
		 &\leq \sum_{i=1\lor j}^n\min\left\{\left|\mE\left[G(X_i)\bigr|\mF_j\right]-\mE\left[G(X_i)|\mF_{j-1}\right]\right|, 2B\right\}\notag\\
		 &\leq\sum_{i=1\lor j}^n\min\left\{\|A^{i-j}\|_{\infty}\mE\left[|\varepsilon_j-\varepsilon_j'|_{\infty}\bigr|\mF_j\right], 2B\right\}\notag\\
		 &\leq \sum_{i=1\lor j}^n\min\left\{p\rho^{-1}\gamma\rho^{(i-j)/\tau}\eta_{j}, 2B\right\},
\end{align}
where $\varepsilon'_j$ is an i.i.d. copy of $\varepsilon_j$ and $\eta_j=\mE\bigr[|\varepsilon_{j1}-\varepsilon_{j1}'|\bigr|\mF_j\bigr].$

Denote $s=\lfloor\tau\log p/\log(1/\rho)\rfloor+1$. Note that $s>0$ is a positive integer. 
For $-s< j\leq0$, we have 
	\begin{align*}
	|L_j|&\leq \sum_{i=0}^{\infty}\min\left\{p\rho^{-1}\gamma\rho^{(i-j)/\tau}\eta_{j}, 2B\right\}\\
		 &\leq \sum_{i=0}^{s-1}\min\left\{p\rho^{-1}\gamma\rho^{(i-j)/\tau}\eta_{j}, 2B\right\}+\sum_{i=s}^\infty\min\left\{p\rho^{-1}\gamma\rho^{(i-j)/\tau}\eta_{j}, 2B\right\}\\
		 &\leq 2sB+\sum_{i=0}^\infty\min\left\{\rho^{-1}\gamma\rho^{i/\tau}\eta_{j}, 2B\right\}
\end{align*}

For $0<j\leq n$, we also have 	
	\begin{align*}
	|L_j|\leq\sum_{i=j}^\infty\min\left\{p\rho^{-1}\gamma\rho^{(i-j)/\tau}\eta_{j}, 2B\right\}
		 \leq -2sB+\sum_{i=0}^\infty\min\left\{\rho^{-1}\gamma\rho^{i/\tau}\eta_{j}, 2B\right\}
\end{align*}
Basic algebra shows that
\begin{align}\label{eq3.4}
	\mE[|L_j|^k|\mathcal{F}_{j-1}]&\overset{(1)}{\leq}\mE\bigg[\Big(2sB+\sum_{i=0}^\infty\min\left\{\rho^{-1}\gamma\rho^{i/\tau}\eta_{j}, 2B\right\}
	\Big)^k\bigg]\notag\\
	&\leq\mE\bigg[2^k\Big((2sB)^k+\Big(\sum_{i=0}^\infty\min\left\{\rho^{-1}\gamma\rho^{i/\tau}\eta_{j}, 2B\right\}\Big)^k\Big)\bigg]\notag\\
	&\leq2^k\bigg[(2sB)^k+\Big(\sum_{i=0}^\infty\Big\|\min\left\{\rho^{-1}\gamma\rho^{i/\tau}\eta_{j}, 2B\right\}\Big\|_k\Big)^k\bigg],
\end{align}
where (1) comes from the independence of $\eta_j$ and $\mathcal{F}_{j-1}$. To analyze \eqref{eq3.4}, we further compute
\begin{align}\label{eq3.5}
	\Big\|\min\left\{\rho^{-1}\gamma\rho^{i/\tau}\eta_{j}, 2B\right\}\Big\|_k&=\Big\|2B\mb{I}\Big(\frac\gamma\rho\rho^{i/\tau}\eta_j\geq2B\Big)+\frac\gamma\rho\rho^{i/\tau}\eta_j\mb{I}\Big(\frac\gamma\rho\rho^{i/\tau}\eta_j\leq2B\Big)\Big\|_k\notag\\
	&\leq2B\bigg(\mP\Big(\frac\gamma\rho\rho^{i/\tau}\eta_j\geq2B\Big)\bigg)^{1/k}+\mE\Big[\Big(\frac\gamma\rho\rho^{i/\tau}\eta_j\Big)^2(2B)^{k-2}\Big]^{1/k}\notag\\
	&\leq\Big(4\sigma^2\frac{\gamma^2}{\rho^2}\Big)^{1/k}\rho^{2i/\tau k}(2B)^{1-2/k}
\end{align}
Plugging \eqref{eq3.5} into \eqref{eq3.4} yields, for some constant $C_1,C_2>0$, that
\begin{align}
	\mE[|L_j|^k|\mathcal{F}_{j-1}]&\leq2^k\bigg[(2sB)^k+4\sigma^2\frac{\gamma^2}{\rho^2}(2B)^{k-2}\Big(\frac1{1-\rho^{2/\tau k}}\Big)^k\bigg]\notag\\
	&\overset{(1)}{\leq}2^k\bigg[(2sB)^k+4\sigma^2\frac{\gamma^2}{\rho^2}(2B)^{k-2}\Big(\frac{\tau k}2\Big)^k\rho^{-2/\tau}\Big(\log(1/\rho)\Big)^k\bigg]\notag\\
	&\overset{(2)}{\leq}2^k\bigg[(2sB)^k+C_1\gamma^2B^{-2}C_2^kB^k\tau^kk!\bigg]\notag\\
	&\leq \gamma^2(Bs\tau)^kk![4+C_1B^{-2}(2C_2)^k]\leq\gamma^2(Bs\tau)^kk!(1+C_1B^{-2})(4+2C_2)^k,
\end{align}
where (1) uses the inequality that $1-x\geq-x\log x$ for $x\in(0,1)$ and (2) uses Stirling formula and the fact that $\rho^{-2/\tau}\leq\rho^{-2}$. Let $\tilde C_1=1+C_1B^{-2}$ and $\tilde C_2=4+2C_2$. Then we obtain
\begin{align}
	\mE\Big[\me^{\lambda L_j}|\mathcal{F}_{j-1}\Big]&\leq1+\sum_{k=2}^\infty\Big[\tilde C_1\gamma^2(\tilde C_2Bs\tau\lambda)^k\Big]=1+\frac{\tilde C_1\gamma^2\tilde C_2^2(Bs\tau)^2\lambda^2}{1-\tilde C_2Bs\tau\lambda}\notag\\
	&\leq\exp\bigg\{\frac{\tilde C_1\gamma^2\tilde C_2^2(Bs\tau)^2\lambda^2}{1-\tilde C_2Bs\tau\lambda}\bigg\}.
\end{align}
Furthermore,
\begin{equation}
	\mE\Big[\exp\Big\{\lambda \sum_{j=s}^nL_j\Big\}\Big]\leq\exp\bigg\{\frac{\tilde C_1\gamma^2\tilde C_2^2(Bs\tau)^2(s+n)\lambda^2}{1-\tilde C_2Bs\tau\lambda}\bigg\}.
\end{equation}
Take $\lambda=x(\tilde C_2Bs\tau x+2\tilde C_1\gamma^2\tilde C_2^2(Bs\tau)^2(s+n))^{-1}$ and by \eqref{P} we have
\begin{align}\label{eq3.9}
	\mP\left(\sum_{j=-s+1}^{n}L_j\geq x\right)\leq \exp\left\{-\frac{x^2}{(1+C_1B^{-2})\gamma^2B^2\tau^4(\log p)^2(\tau\log p+n)+C_4\tau^2 B(\log p)x}\right\}.
\end{align}
Similarly, for $j\leq-s$, since $p\leq\rho^{-s/\tau}$,
\begin{align*}
	|L_j|
		 &\leq \sum_{i=0}^\infty\min\left\{\rho^{-1}\gamma\rho^{(i-j-s)/\tau}\eta_{j}, 2B\right\}.
\end{align*}
By the same argument, we immediate have
\begin{align}\label{eq3.10}\mP\bigg(\sum_{j=-\infty}^{-s}L_j\geq x\bigg)
	\leq\exp\left\{-\frac{x^2}{C_3\gamma^2\tau^3+C_4\tau B x}\right\},\end{align}
where $C_3=32\me^2\sigma^2(2\pi)^{-1/2}[\rho^2\log(1/\rho)]^{-3}$ and $C_4=8\me[\log(1/\rho)]^{-1}$.
By \eqref{eq3.9}, \eqref{eq3.10} and symmetrization argument, we complete the proof.
\end{proof}

\begin{proof}[Proof of Corollary \ref{cor3.2}]
	It follows from the proof of lemma \ref{lem3.1} without any extra technical difficulty.
\end{proof}

\begin{proof}[Proof of Theorem \ref{thm3.3}]
	Define $G_{jk}:\mb{R}^p\to\mb{R}$ be defined as $G_{jk}(x)=\big(xx^\top w(x)\big)_{jk}=x_jx_kw(x)$ for $j,k=1,\dots,p$, and hence $|G(x)|\leq T$. Let $u(x)=w^{1/3}(x)$. Observe that
	\begin{align*}
		|G_{jk}(x)-G_{jk}(y)|&\leq |x_ju(x)\,x_ku(x)-y_ju(y)\,y_ku(y)|u(x)
		+|y_ju(y)\,y_ku(y)||u(x)-u(y)|\\
		&\leq |x_iu(x)-y_iu(y)||x_ju(x)|+|x_ju(x)-y_ju(y)||y_iu(y)|+T^2|u(x)-u(y)|\\
		&\leq 3T|x-y|_\infty.
	\end{align*}
	By lemma \ref{lem3.1}, take $x=c\gamma\tau^2Tn^{-1/2}(\log p)^{3/2}$
	$$\mP\bigg(|\widehat\Sigma_{x,jk}-\Sigma_{x,jk}|\geq cTx\bigg)=\mP\bigg(\Big|\frac1n\sum_{i=1}^nG_{jk}(X_{i-1})-\frac1n\sum_{i=1}^n\mE G_{jk}(X_{i-1})\Big|\geq cTx\bigg)\leq 4p^{-c_1}.$$
	A union bound yields
		$$\mP\bigg(\|\widehat\Sigma_{x}-\Sigma_{x}\|_{\max}\geq cTx\bigg)\leq 4p^{-c_0},$$
	where $c_0=c_1-1>0$.
\end{proof}

\begin{proof}[Proof of Theorem \ref{thm3.4}]
	Denote $H=\nabla^2 L_n(\beta^*)$. First we write 
	$$\|H-\Sigma\|_{\max}=\max_{1\leq k\leq p}\Big\|\frac{1}{n}\sum_{i=1}^n\psi'(\varepsilon_{ik})X_{i-1}X_{i-1}^\top w(X_{i-1})-\mE[\psi'(\varepsilon_{ik})X_{i-1}X_{i-1}^\top w(X_{i-1})]\Big\|_{\max}.$$
	Using Corollary \ref{cor3.2}, it follows from the same argument of the proof of Theoremt \ref{thm3.3} that for some constant $c_1>1$, with probability at least $1-4p^{-c_1}$
	$$\max_{1\leq k\leq p}\Big\|\frac{1}{n}\sum_{i=1}^n\psi'(\varepsilon_{ik})X_{i-1}X_{i-1}^\top w(X_{i-1})-\mE[\psi'(\varepsilon_{ik})X_{i-1}X_{i-1}^\top w(X_{i-1})]\Big\|_{\max}\lesssim \gamma\tau^2T^2 n^{-1/2}(\log p)^{3/2},$$
	Finally, a union bound over $1\leq k\leq p$ yields the conclusion.
\end{proof}

\begin{proof}[Proof of Lemma \ref{lem1.1}]
	The strategy is to consider each component of $\widehat\mu$ and take a union bound.  Observe that
	$$|\widehat\mu_k-\mu_k|\leq \bigg|\frac1n\sum_{i=1}^n\psi'(\widehat\varepsilon_{ik})-\mE[\psi'(\widehat\varepsilon_{ik})]\bigg|+|\mE\psi'(\widehat\varepsilon_{ik})-\mE\psi'(\varepsilon_{ik})|, \quad k=1,2,\dots, p.$$
	Since $|\psi''|$ is bounded, by the mean value theorem, we have that for some $\xi$ between $x$ and $y$, $$|\psi'(x)-\psi'(y)|=|\psi''(\xi)(x-y)|\lesssim|x-y|.$$So 
	it can be verified that $\psi'(X_{ik}-X_{i-1}^\top\widehat\beta_k)$ satisfies the conditions in Corollary 2.5 of \cite{liu2021robust}. 
	By Corollary 2.5 of \cite{liu2021robust}, it holds that
	$$\bigg|\frac1n\sum_{i=1}^n\psi'(\widehat\varepsilon_{ik})-\mE[\psi'(\widehat\varepsilon_{ik})]\bigg|\lesssim\gamma\tau^2\sqrt{\frac{\log p}n}$$
	with probability at least $1-2p^{-c}$ for some positive constant $c$.
	Moreover,
	$$\max_{1\leq k\leq p}|\mE\widehat\varepsilon_{ik}-\mE\varepsilon_{ik}|\lesssim\max_{1\leq k\leq p}\mE\big[|X_{i-1}^\top(\widehat\beta_k-\beta^*)|\big]\lesssim |\widehat\beta-\beta^*|_1,$$
	where the last inequality comes from the fact that $X_{i-1}$ has bounded second moment.
\end{proof}

\section{Proofs of Result in Section \ref{GA}}
Before proving Theorem \ref{thm5.4}, we will first state and prove the corresponding lemmas in the outline listed at the end of section \ref{GA}.

\begin{lem}\label{lem4.2}
Suppose $\mE[\varepsilon_{ik}^2]\leq\sigma^2$ for all $1\leq k\leq p$ and the odd function $\psi(\cdot)$ satisfies that $|\psi(\cdot)|\leq C$ and $|\psi'(\cdot)|\leq C$, then we have
$$\mP\bigg(\big|({T_Y-T_{Y,m})}/{\sqrt n}\big|_\infty\geq x\bigg)\leq \frac{c_1sp^3\gamma^2\rho^{m/\tau}}{x^2}=:f_1(x, m),$$	
for some constants $C_1,C_2>0$.
\end{lem}
\begin{proof}[Proof of Lemma \ref{lem4.2}]
	Let $D_i=Y_i-Y_{i,m}$. For any $\lambda>0$, by Markov inequality we have
	\begin{equation}\label{mi}
		\mP\bigg(\sum_{i=1}^nD_{ij}/\sqrt n\geq x\bigg)\leq\frac{\mE\big[\big(\sum_{i=1}^nD_{ij}/\sqrt n\big)^2\big]}{x^2}.
	\end{equation}
	Notice that the martingale difference $\{D_{ij}\}_{i=1}^n$ satisfies
	$$|D_{ij}|\lesssim \sqrt s|X_i-X_{i,m}|_\infty.$$
	Thus, 
		$$\|D_{ij}\|_2\leq \||X_i-X_{i,m}|_\infty\|_2\leq \sum_{l=m+1}^\infty\|A^l\|_\infty\||\varepsilon_{i-l}|_\infty\|_2\lesssim \sqrt sp\gamma\rho^{m/\tau}.$$
	By Burkholder inequality (\cite{burkholder1973distribution}), we have
	\begin{align}
		\mE\bigg[\bigg(\sum_{i=1}^n D_{ij}/\sqrt n\bigg)^2\bigg]&\lesssim\mE[|D_{ij}|^2]\lesssim sp^2\gamma^2\rho^{2m/\tau}
	\end{align}
	Hence, by \eqref{mi},
	$$\mP\bigg(\sum_{i=1}^nD_{ij}/\sqrt n\geq x\bigg)\leq \frac{c_1'sp^2\gamma^2\rho^{m/\tau}}{x^2}$$

	Finally, symmetrization and a union bound give the desired result.
\end{proof}

\begin{lem}\label{lem4.1}
	Under the assumptions in Lemma \ref{lem4.2}, it holds that
	$$\mP\big(|T_{Y,S}|_\infty/\sqrt n\geq x\big)\leq 2p\exp\bigg\{-\frac{nx^2}{C_1\sqrt sT\sqrt nx+C_2mwsT^2\sigma^2}\bigg\}=:f_2(x,m).$$
\end{lem}
\begin{proof}[Proof of Lemma \ref{lem4.1}]
	By the property of $\psi(\cdot)$ and the mean value theorem, we have $|\psi(x)|\leq C|x|$.
	Consider the first coordinate $(T_{Y,S})_1$ of $T_{Y,S}$.
	We can write  $(T_{Y,S})_1=\sum_{i=j_1}^{j_{r}}Y_{i,m,1}$, where $r=m\omega$. Observe that $\{Y_{i,m,1}\}$ is a martingale difference adapted to the filtration $\{\mathcal{F}_i=\sigma(\varepsilon_i, \varepsilon_{i-1},
	\dots)\}$ and that $|Y_{i,m,1}|\leq \psi(\varepsilon_{ik})\sqrt sT\leq C\sqrt sT$. We shall establish a Bernstein-type inequality for the sum of martingale differences $(T_{Y,S})_1$:
	\begin{equation}\label{eq5.1}
		\mP((T_{Y,S})_1\geq x)\leq \me^{-\lambda x}\mE\me^{\lambda\sum_{i=j_1}^{j_r}Y_{i,m,1}}, \quad\text{for any }\lambda > 0.
	\end{equation}
	We now bound $\mE\me^{\lambda\sum_{i=j_1}^{j_r}Y_{i,m,1}}$ from above. By the tower property,
	\begin{align}\label{eq5.2}
		\mE\exp\Big\{\lambda\sum_{i=j_1}^{j_r}Y_{i,m,1}\Big\}&=\mE\bigg[\mE\Big[\exp\Big\{\lambda\sum_{i=j_1}^{j_r}Y_{i,m,1}\Big\}\Big|\mathcal{F}_{r-1}\Big]\bigg]\notag\\
		&=\mE\bigg[\exp\Big\{\lambda\sum_{i=j_1}^{j_{r-1}}Y_{i,m,1}\Big\} \mE[\me^{\lambda Y_{j_r, m,1}}|\mathcal{F}_{j_r-1}]\bigg]
	\end{align}
	Now, consider 
	\begin{align}\label{eq5.3}
		\mE[\me^{\lambda Y_{j_r, m,1}}|\mathcal{F}_{j_r-1}]&=1+\mE\Big[\sum_{t=2}^\infty\frac{(\lambda Y_{j_r,m,1})^t}{t!}\Big|\mathcal{F}_{j_r-1}\Big]
		\leq1+\mE\Big[\lambda^2T^2s\psi^2(\varepsilon_{j_rk})\sum_{t=0}^\infty(\lambda T\sqrt s C)^t\Big]\notag\\
		&\overset{\text{(1)}}{\leq}1+\frac{C\lambda^2T^2s\sigma^2}{1-C\lambda T\sqrt s}
		\leq\exp\bigg\{\frac{C\lambda^2T^2s\sigma^2}{1-C\lambda T\sqrt s}\bigg\}
	\end{align}
	where the inequality (1) makes use of the fact that $\psi^2(\varepsilon_{j_r,k})\leq\varepsilon_{j_r,k}^2$.
	Plug \eqref{eq5.3} into \eqref{eq5.2} and we obtain
	\begin{equation}
		\mE\exp\Big\{\lambda\sum_{i=j_1}^{j_r}Y_{i,m,1}\Big\}\leq\exp\bigg\{\frac{C\lambda^2T^2s\sigma^2}{1-C\lambda T\sqrt s}\bigg\}\mE\bigg[\exp\Big\{\lambda\sum_{i=j_1}^{j_{r-1}}Y_{i,m,1}\Big\}\bigg]
	\end{equation}
	Iterating this procedure yields
	\begin{equation}
		\mE\exp\Big\{\lambda\sum_{i=j_1}^{j_r}Y_{i,m,1}\Big\}\leq\exp\bigg\{\frac{Cm\omega\lambda^2T^2s\sigma^2}{1-C\lambda T\sqrt s}\bigg\}
	\end{equation}
	Choose $\lambda=x(CT\sqrt s+2CmwT^2s\sigma^2)^{-1}$ and by \eqref{eq5.1} we have
	$$\mP((T_{Y,S})_1\geq x)\leq \exp\bigg\{-\frac{x^2}{C_1T\sqrt sx+C_2m\omega T^2s\sigma^2}\bigg\}.$$
	The symmetrization argument and a union bound deliver the desired result.
\end{proof}

\begin{lem}\label{lem4.3}
	Suppose the scaling condition holds that $sT^2(\log(pn))^7/n\leq c_3n^{-c_4}$. Assume that $\mE[X_{ik}]\leq C'$ for all $1\leq k \leq p$. Then we have the following Gaussian Approximation result that
	$$\mathcal{U}:=\sup_{t\in\mb{R}}\bigg|\mP\big(|T_{Y,L}/\sqrt n|_\infty\leq t\big)-\mP\big(|Z|_\infty\leq t\big)\bigg|\leq cn^{-c'}$$
	for some constants $c,c'>0$.
\end{lem}
\begin{proof}[Proof of Lemma \ref{lem4.3}]
	Recall that $\xi_b=\sum_{i\in L_b}Y_{i,m}/\sqrt M$, thus 
	$$\mathcal{U}=\sup_{t\in\mb{R}}\bigg|\mP\big(|\frac1{\sqrt w}\sum_{b=1}^w \xi_b|_\infty\leq t\big)-\mP\big(|Z|_\infty\leq t\big)\bigg|.$$
	Observe that $\xi_1,\xi_2,\dots,\xi_w$ are independent random variables. We shall apply Corollary 2.1 of \cite{chernozhukov2013gaussian} by verifying the condition (E.1) therein. For completeness, the conditions are stated below.
	\begin{itemize}
		\item[(i)] $c_1\leq \mE[\xi_{bj}^2]\leq c_2$ for all $1\leq j\leq p$.
		\item[(ii)] $\max_{k=1,2}\mE[|\xi_{bj}|^{2+k}/B_n^k]+\mE[\exp(|\xi_{bj}/B_n|)]\leq4,$ for some $B_n>0$ and all $1\leq j
		\leq p$.
		\item[(iii)] $B_n^2(\log(pn))^7/n\leq c_3n^{-c_4}$.
	\end{itemize}
	To verify condition (i), we see that
	$$\mE[\xi_{bj}^2]\leq c\,\sigma^2\mE[\mE[\Omega_{x,j}^\top X_{i}w(X_i)|\varepsilon_{i-m},\dots,\varepsilon_i]^2]\leq c\,\mE[(\Omega_{x,j}^\top X_iw(X_i))^2]\leq c\,\Omega_{x,j}^\top\Sigma_x\Omega_{x,j}\leq c\,\Omega_{x,jj},$$
	where $\Omega_{x,j}$ is the $j$-th row of $\Omega_x$ and $\Omega_{x,jj}$ is the $j$-th diagonal entry of $\Omega_x$.
	Now we check condition (ii). By Theorem 3.2 of \cite{burkholder1973distribution}, we have for $k\geq 2$,
	$$\mE[|\xi_{bj}|^k]\leq 18k^k\mE\big[\big|Y_{ij,m}\big|^k\big]\lesssim k!\me^k\mE\big[\big|Y_{ij,m}\big|^2\big](\sqrt s T)^{k-2}\lesssim k!\me^k(\sqrt s T)^{k-2}.$$
	Therefore, take $B_n=C\sqrt sT$ for sufficiently large $C>0$ and we have 
	$$\mE[\exp(|\xi_{bj}/B_n|)]\leq1+C_1\sum_{k=1}^\infty(\me/C)^k<2.$$
	Moreover, for a suitable choice of $C>0$,
	$$\max_{k=1,2}\mE[|\xi_{bj}|^{2+k}/B_n^k]<2.$$
	Hence, condition (ii) is satisfied. Condition (iii) is guaranteed by the scaling assumption.
\end{proof}
Now, we are ready to give the proof of Theorem \ref{thm5.4}.
\begin{proof}[Proof of Theorem \ref{thm5.4}]
	By triangle inequality,
	\begin{align}\label{H}
	\mathcal{H}&\leq \sup_{t\in\mb{R}}\bigg|\mP\big(|T_{Y}/\sqrt n|_\infty\leq t\big)-\mP\big(|T_{Y,L}/\sqrt n|_\infty\leq t\big)\bigg|+\sup_{t\in\mb{R}}\bigg|\mP\big(|T_{Y,L}/\sqrt n|_\infty\leq t\big)-\mP\big(|Z|_\infty\leq t\big)\bigg|\notag\\
	&=:I+II.
	\end{align}
	For any $\eta>0$, elementary calculation shows that
	\begin{align}
	I&\leq \mP\big(\big|(T_Y-T_{Y,L})/\sqrt n\big|_\infty>\eta\big)+\sup_{t\in\mb{R}}\mP\bigg(\bigg|\big|T_{Y,L}/\sqrt n\big|_\infty-t\bigg|\leq\eta\bigg)\notag\\
	&\leq \mP\big(\big|(T_Y-T_{Y,m})/\sqrt n\big|_\infty>\frac\eta2\big)+\mP\big(\big|T_{Y,S}/\sqrt n\big|_\infty>\frac\eta2\big)+\sup_{t\in\mb{R}}\mP\bigg(\bigg|\big|T_{Y,L}/\sqrt n\big|_\infty-t\bigg|\leq\eta\bigg)
	\end{align}	
By lemma \ref{lem4.1} and \ref{lem4.2},
\begin{align}\label{T_Y}
\mP\big(\big|(T_Y-T_{Y,m})/\sqrt n\big|_\infty>\frac\eta2\big)\leq f_1(\eta/2,m)\quad\text{and}\quad
 \mP\big(\big|T_{Y,S}/\sqrt n\big|_\infty>\frac\eta2\big)\leq f_2(\eta/2).
 \end{align}
By lemma \ref{lem4.3} and theorem 3 of \cite{chernozhukov2015comparison}, we obtain that
	\begin{align}\label{T_YL}
		\sup_{t\in\mb{R}}\mP\bigg(\bigg|\big|T_{Y,L}/\sqrt n\big|_\infty-t\bigg|\leq\eta\bigg)&\leq \sup_{t\in\mb{R}}\mP\bigg(\bigg|\big|Z\big|_\infty-t\bigg|\leq\eta\bigg)+\mathcal{U}\notag\\
		&\lesssim \eta\sqrt{\log p}+\eta\sqrt{\log(1/\eta)}+cn^{-c'},
	\end{align}
and that
	\begin{equation}\label{II}
		II=\mathcal{U}\leq cn^{-c'}.
	\end{equation}
By \eqref{H}, \eqref{T_Y}, \eqref{T_YL} and \eqref{II}, we obtain the inequality stated in the theorem.
\end{proof}


\bibliographystyle{apalike}	 
\begin{bibliography}{inference}
\end{bibliography}

\end{document}